\documentclass[a4paper,12pt]{article}
\usepackage{amsfonts}
\usepackage{amsmath}
\usepackage{amssymb}
\usepackage{bbm}
\usepackage[utf8]{inputenc}
\usepackage{color}
\usepackage{xcolor}
\newtheorem{theorem}{Theorem}

\newtheorem{corollary}[theorem]{Corollary}

\newtheorem{definition}[theorem]{Definition}
\newtheorem{example}[theorem]{Example}

\newtheorem{lemma}[theorem]{Lemma}

\newtheorem{proposition}[theorem]{Proposition}

\newenvironment{proof}[1][Proof]{\textbf{#1.} }{\ \rule{0.5em}{0.5em}}
\scrollmode

\newcommand{\re}{\mathfrak{R}\kern-0.5truept e}
\newcommand{\im}{\mathfrak{I}\kern-1.0truept m}

\newcommand{\unit}{\hbox{\rm 1\kern-2.8truept l}}
\newcommand{\T}{\mathcal{T}}

\newcommand{\me}{\hbox{\rm e}}
\newcommand{\mi}{\mathrm{i}}

\newcommand{\modulo}[1]{\left| #1 \right|}
\newcommand{\rescalar}[2]{\Re \scalare{#1}{#2}}
\newcommand{\scalare}[2]{\left\langle #1 , #2 \right\rangle}
\newcommand{\scalar}[2]{\scalare{#1}{#2}}

\newcommand{\comm}[2]{\left[ #1 , #2 \right]}
\newcommand{\conj}[1]{\overline{#1}}

\newcommand{\norm}[1]{\left\| #1 \right\|}
\font\rm=cmr12

\font\rmseven=cmr8

\begin{document}
\title{The decoherence-free subalgebra of Gaussian Quantum Markov Semigroups}
\author{J. Agredo, F. Fagnola, D. Poletti}


\date{ }
\maketitle


\begin{abstract}
We demonstrate a method for finding the decoherence-subalgebra $\mathcal{N}(\T)$ of a Gaussian quantum Markov semigroup on
the von Neumann algebra $\mathcal{B}(\Gamma(\mathbb{C}^d))$ of all bounded operator on the Fock space $\Gamma(\mathbb{C}^d)$
on $\mathbb{C}^d$. We show that $\mathcal{N}(\T)$ is a type I von Neumann algebra
$L^\infty(\mathbb{R}^{d_c};\mathbb{C})\overline{\otimes}\mathcal{B}(\Gamma(\mathbb{C}^{d_f}))$ determined, up to unitary
equivalence, by two natural numbers $d_c,d_f\leq d$. This result is illustrated by some applications and
examples.

\end{abstract}


\section{Introduction}

Quantum channels and quantum Markov semigroups (QMS) describe the evolution of an open quantum system
subject to noise because of the interaction with the surrounding environment.
Couplings to external degrees of freedom typically lead to decoherence.
Gaussian quantum channels and Markov semigroups play a key role because several models are
based on linear couplings of bosonic systems to other bosonic systems with quadratic Hamiltonians.
As a result, the time evolution is then determined by a Gaussian channel in the discrete time
case and Markov semigroup in the time continuous case.

Decoherence-free subalgebras determine observables whose evolution is not affected by noise and play a fundamental
role not only in the analysis of decoherence (see \cite{AFR,BlOl,CaSaUm,Isar,Lidar,LiChWh} and the references
therein) but also in the study of the structure of QMSs (see \cite{FaSaUm}).

The case of a norm-continuous QMS $(\mathcal{T}_t)_{t\geq 0}$ has been
extensively studied (\cite{DFSU,FaSaUm} and the references therein). The generator is represented
in a GKLS form $\mathcal{L}(x)=\mi [H,x]-(1/2)\sum_{\ell}\left(L_\ell^*L_\ell x - 2L_\ell^*L_\ell x
+ x L_\ell^*L_\ell  \right)$ for bounded operators $L_\ell, H$ with $H$ self-adjoint and the
decoherence-free subalgebra is characterized as the commutant of operators $L_\ell,L_\ell^*$
and their iterated commutators with $H$ (\cite{DFSU} Proposition 2.3).

Gaussian QMSs arise in several relevant models and form a class with a rich
structure with a number of explicit formulas (\cite{AFP,CrFi,GJN,Te}), yet they are not norm-continuous.
However, it is known that their generators can be written in a generalized GKLS form
with operators $L_\ell$ that are linear (see (\ref{eq:Lell})) and $H$ quadratic (see (\ref{eq:H}))
in boson creation and annihilation operators $a_j,a^\dagger_k$ (\cite{Demoen,Po2021}).
As a consequence, operators $L_\ell,L_\ell^*$ and their iterated generalized commutators
with $H$ are linear in boson creation and annihilation operators $a_j,a^\dagger_k$.

In this paper we consider Gaussian QMS on the von Neumann algebra
$\mathcal{B}(\Gamma(\mathbb{C}^d))$ of all bounded operators on the Boson Fock space on
$\mathbb{C}^d$ and characterize their decoherence-free subalgebras as  generalized commutants
of these iterated generalized commutators (Theorem \ref{th:NT-comm}).
Indeed, we show that it suffices to consider iterated commutators up to the order $2d-1$.
Moreover, we prove (Theorem \ref{th:NT-Gauss}) that the decoherence-free subalgebra is a type I von Neumann
algebra $L^\infty(\mathbb{R}^{d_c};\mathbb{C})\overline{\otimes}\mathcal{B}(\Gamma(\mathbb{C}^{d_f}))$
determined, up to unitary equivalence, by two natural numbers $d_c,d_f\leq d$.
This conclusion is illustrated by some examples and a detailed analysis of the case
of a Gaussian QMS with a single operator $L_\ell$.

The symplectic structure on $\mathbb{C}^d$ plays a fundamental role in
the origin of von Neumann algebras of the type $L^\infty(\mathbb{R}^{d_c};\mathbb{C})$ as possible decoherence-free
subalgebras determined by generalized commutants via Araki's duality theorem.
Moreover, even if at a purely algebraic level Theorem \ref{th:NT-comm} looks like a natural
generalization of the norm continuous case, several difficulties arise from unboundedness of operators $L_\ell$ and $H$
and, as a consequence, unboundedness of the generator $\mathcal{L}$ of the QMS $\mathcal{T}$.
The defining property (\ref{eq-NT}) of $\mathcal{N}(\mathcal{T})$ involves an operator $x$ and the product $x^*x$,
but, even if $x$ belongs to the domain of the generator of the QMS $\mathcal{T}$, there is no reason why
$x^*x$ should (see e.g. \cite{FF00}) therefore one has to work with quadratic forms.
Domain problems arising from generalized commutators have to be carefully handled because
one needs to make sense of generalized commutations such as $x[H,L_\ell]\subseteq [H,L_\ell]x$ and
$\hbox{\rm e}^{\mi t H}x\hbox{\rm e}^{-\mi t H} L_\ell \subseteq  L_\ell\, \hbox{\rm e}^{\mi t H}x\hbox{\rm e}^{-\mi t H}$
for $x\in \mathcal{N}(\mathcal{T})$.

The decoherence-free subalgebra of a QMS with unbounded generator
was also characterized in \cite{DhFaRe} with several technical assumptions that hold
in the case of Gaussian semigroups and using a dilation of the QMS via quantum stochastic calculus.
The proof we give here (Appendix B) is simpler because it does not appeal to these assumptions
and is more direct because it does not use quantum stochastic calculus.

The paper is organized as follows. In Section \ref{sec:definition} we introduce Gaussian QMS,
present their construction by the minimal semigroup method (see \cite{Fa}), well-definedness
(conservativity or identity preservation) applying the sufficient condition of \cite{ChFa}
and the explicit formula for the action on Weyl operators (Theorem \ref{th:explWeyl}).
Proofs, that can be obtained from applications of standard methods, are collected in
Appendix A. Section \ref{sect:NT} contains the main results of the paper. We first recall the
definition of decoherence-free subalgebra. Then we prove its characterization for a Gaussian QMS
(Theorem \ref{th:NT-comm} with proof in Appendix B). Finally we prove the structure result
Theorem \ref{th:NT-Gauss}.  Applications and examples are presented in Section \ref{sect:appl}.

\section{Gaussian QMSs} \label{sec:definition}
In this section we introduce our class of Gaussian semigroups
starting from their generators and fix some notation.
Let $\mathsf{h}$ be the Fock space $\mathsf{h}=\Gamma(\mathbb{C}^d)$ which is
isometrically isomorphic to $\Gamma(\mathbb{C})\otimes\cdots\otimes\Gamma(\mathbb{C})$ with canonical
orthonormal basis $(e(n_1,\ldots,n_d))_{n_1,\ldots,n_d\geq0}$
(with $e(n_1,\ldots,n_d)=e_{n_1}\otimes\ldots\otimes e_{n_d}$). Let $a_j, a_j^{\dagger}$ be the creation and
annihilation operator of the Fock representation of the $d$-dimensional Canonical Commutation Relations (CCR)
\begin{eqnarray*}
a_j\,e(n_1,\ldots,n_d)& =& \sqrt{n_j}\ e(n_1,\ldots,n_{j-1},n_j-1,\ldots,n_d),\\
a_j^{\dagger}\,e(n_1,\ldots,n_d)&=&\sqrt{n_j+1}\ e(n_1,\ldots,n_{j-1},n_j+1,\ldots,n_d),
\end{eqnarray*}
The CCRs are written as $[a_j,a_k^{\dagger}]=\delta_{jk}\mathbbm{1}$, where $[\cdot,\cdot]$ denotes the
commutator, or, more precisely, $[a_j,a_k^{\dagger}]\subseteq \delta_{jk}\mathbbm{1}$ because the
domain of the operator in the left-hand side is smaller.

For any $g\in\mathbb{C}^d$, define the exponential vector $e_g$ associated with $g$ by
 \[
 e_g=\sum_{n\in \mathbb{N}^d}\frac{g_1^{n_1}\cdots g_d^{n_d}}{\sqrt{n_1!\cdots n_d!}}\,e(n_1,\ldots,n_d)
 \]
Creation and annihilation operators with test vector $v\in\mathbb{C}^d$ can also be defined on the
total set of exponential vectors (see \cite{Pa}) by
 \[
 a(v)e_g=\langle v,g \rangle e_g, \quad a^{\dagger}(v)e_g=\frac{\hbox{\rm d}}{\hbox{\rm d}\varepsilon}e_{g+\varepsilon u}|_ {\varepsilon=0}
 \]
for all $u\in\mathbb{C}^d$. The unitary correspondence $\Gamma(\mathbb{C}^d)\mapsto\Gamma(\mathbb{C})\otimes\cdots
\otimes\Gamma(\mathbb{C})$
\[
e_g\mapsto\sum\limits_{n_1\geq0,\ldots, n_d\geq0}\frac {g_1^{n_1}\ldots g_d^{n_d}}{\sqrt{n_1!\ldots n_d!}}
\,e_{n_1}\otimes\ldots\otimes e_{n_d}
\]
allows one to establish the identities
 \[
 a(v)=\sum\limits_{j=1}^d\overline{v}_ja_j, \quad a^{\dagger}(u)=\sum\limits_{j=1}^d u_ja_j^{\dagger}
 \]
for all $u^{\hbox{\rmseven T}}=[u_1,\ldots,u_d],v^{\hbox{\rmseven T}}=[v_1,\ldots,v_d]\in\mathbb{C}^d$.

The above operators are obviously defined on the linear manifold $D$ spanned by the
elements $(e(n_1,\ldots,n_d))_{n_1,\ldots,n_d\geq0}$ of the canonical orthonormal basis of $\mathsf{h}$ that
turns out to be an essential domain for all the operators considered so far. This also happens for quadrature
operators
\begin{equation}\label{eq:quadrature}
q(u) = \left(a(u)+a^\dagger(u)\right)/\sqrt{2}\qquad u\in\mathbb{C}^d
\end{equation}
that are symmetric and essentially self-adjoint on the domain $D$ by Nelson's theorem
on analytic vectors (\cite{ReSi} Th. X.39 p. 202). The linear span of exponential vectors also
turns out to be an essential domain for operators $q(u)$ for the same reason. If the vector $u$ has real
(resp. purely imaginary) components one finds position (resp. momentum) operators and the commutation relation
$[q(u),q(v)]\subseteq \mi \Im \langle u, v\rangle \unit$ (where $\Im$ and $\Re$ denote the imaginary and real part
of a complex number). Momentum operators, i.e. quadratures $q(\mi r)$ with $r\in\mathbb{R}^d$ are also denoted
by $p(r)=\sum_{1\leq j\leq d} r_j p_j$ where $p_j=\mi(a^\dagger_j-a_j)/\sqrt{2}$. In a similar way we write
$q(r)=\sum_{1\leq j\leq d} r_j q_j$ with $q_j=q(e_j)=(a^\dagger_j+a_j)/\sqrt{2}$.

Another set of operators that will play
an important role in this paper are the Weyl operators, defined on the exponential vectors via the formula
\[
W(z)e_g=\hbox{\rm e}^{-\parallel z\parallel^2/2-\langle z,g \rangle}e_{z+g}\quad z,g\in\mathbb{C}^d.
\]
By this definition $\langle W(z)e_f,W(z)e_g\rangle = \langle e_f,e_g\rangle$ for all $f,g\in\mathbb{C}^d$,
therefore $W(z)$ extends uniquely to a unitary operator on $\mathsf{h}$. Weyl operators satisfy the CCR in the
exponential form, namely, for every $z,z^\prime \in \mathbb{C}^d$,
\begin{equation} \label{eq:WeylCCR}
	W(z)W(z^\prime) = \me^{-\mi \Im \scalar{z}{z^\prime}} W(z+z^\prime).
\end{equation}
It is well-known that $W(z)$ is the exponential of the anti self-adjoint operator $-\mi\sqrt{2}\, q(\mi z)$
\begin{equation} \label{eq:weylQ}
W(z) = \hbox{\rm e}^{-\mi\sqrt{2}\, q(\mi z)} = \hbox{\rm e}^{z a^\dagger-\overline{z}a}.
\end{equation}
Finally, we recall here two relevant properties that are valid on $D$ and on suitable dense domains
\begin{equation}\label{eq:WzCa}
 \left[\,a(v),W(z)\,\right] = \langle v, z\rangle W(z), \qquad
\left[\,a^\dagger (u),W(z)\,\right] =\langle z, u\rangle W(z).
\end{equation}

A QMS $\mathcal{T}=(\mathcal{T}_t)_{t\geq 0}$ is a weakly$^*$-continuous semigroup of completely positive,
identity preserving, weakly$^*$-continuous maps on $\mathcal{B}(\mathsf{h})$.
The predual semigroup $\mathcal{T}_*= (\mathcal{T}_{*t})_{t\geq 0}$ on the predual space of trace class
operators on $\mathsf{h}$ is a strongly continuous contraction semigroup.

Gaussian QMSs can be defined in various equivalent ways. Here we introduced them through their generator
because it is the object we are mostly concerned with.
The pre-generator, or form generator, of a Gaussian QMSs can be represented in a generalized (since operators
$L_\ell, H$ are unbounded) Gorini–Kossakowski–Lindblad-Sudarshan (GKLS) form (see \cite{Po2021} Theorems 5.1, 5.2
and also  \cite{Demoen,Vheu})
\begin{equation}\label{eq:GKLS}
\mathcal{L}(x) = \mi\left[ H, x\right]
-\frac{1}{2}\sum_{\ell=1}^m \left( L_\ell ^*L_\ell\, x - 2 L_\ell^* x L_\ell + x\, L_\ell ^*L_\ell\right).
\end{equation}
where $1 \leq m \leq 2d$, and
	\begin{align}
		H&= \sum_{j,k=1}^d \left( \Omega_{jk} a_j^\dagger a_k + \frac{\kappa_{jk}}{2} a_j^\dagger a_k^\dagger + \frac{\overline{\kappa_{jk}}}{2} a_ja_k \right) + \sum_{j=1}^d \left( \frac{\zeta_j}{2}a_j^\dagger + \frac{\bar{\zeta_j}}{2} a_j \right), \label{eq:H}\\
		 L_\ell &= \sum_{k=1}^d \left( \overline{v_{\ell k}} a_k + u_{\ell k}a_k^\dagger\right)
= a( v_{\ell \bullet}) + a^\dagger ( u_{\ell \bullet}), \label{eq:Lell}
	\end{align}
$\Omega:=(\Omega_{jk})_{1\leq j,k\leq d} = \Omega^*$ and $\kappa:= (\kappa_{jk})_{1\leq j,k\leq d}= \kappa^{\hbox{\rmseven T}}
\in M_d(\mathbb{C})$, are $d\times d$ complex matrices with $\Omega$ Hermitian and $\kappa$ symmetric,
$V=(v_{\ell k})_{1\leq \ell\leq m, 1\leq  k\leq d}, U=(u_{\ell k})_{1\leq \ell\leq m, 1\leq  k\leq d} \in M_{m\times d}(\mathbb{C})$
are $m\times d$ matrices and $\zeta=(\zeta_j)_{1\leq j\leq d} \in \mathbb{C}^d$. The notation $v_{\ell \bullet}$ and $u_{\ell \bullet}$
refers to vectors in $\mathbb{C}^d$ obtained from the $\ell$-th row of the corresponding matrices.

We exclude the case where the pre-generator $\mathcal{L}$ reduces to
the Hamiltonian part $\mi[H,x]$ and so we suppose that one among matrices $V,U$ is non-zero. An application of Nelson's theorem
on analytic vectors (\cite{ReSi} Th. X.39 p. 202) shows that $H$, as an operator with domain $D$, is essentially selfadjoint.
In addition, operators $L_\ell$ are closable therefore we will identify them with their closure.

It can be shown (see \cite{Po2021} Theorems 5.1, 5.2) that a QMS  $\mathcal{T}$ is Gaussian if
maps $\mathcal{T}_{*t}$ of the predual semigroup $\mathcal{T}_*$ preserve Gaussian states or, still in an
equivalent way, maps $\mathcal{T}_t$ act explicitly on Weyl operators (Theorem \ref{th:explWeyl} below).

Clearly, $\mathcal{L}$ is well defined on the dense (not closed) sub-$^*$-algebra of $\mathcal{B}(\mathsf{h})$ generated by rank one
operators $|\xi\rangle\langle \xi'|$ with $\xi,\xi'\in D$ because all operator compositions make sense. However, since the operators
$H,L_\ell$ are unbounded, the domain of $\mathcal{L}$ is not the whole of $\mathcal{B}(\mathsf{h})$. For this reason we look
at it as a pre-generator and describe in detail its extension to a generator of a QMS by the minimal semigroup method
(Theorem \ref{th:G-QMS!} below).

\smallskip
\noindent{\bf Remark.} The above generalized GKLS form is the most general with operators $L_\ell$ which are first order
polynomials in $a_j,a^\dagger_j$ and the self-adjoint operator $H$ which is a second order polynomial in $a_j,a^\dagger_j$.
Indeed, in the case where $L_\ell$ are as above plus a multiple of the identity operator, exploiting non uniqueness of GKLS
representations (see \cite{Pa}, section 30) one can always apply a translation  and reduce himself to the previous case.

\smallskip
We choose the minimum number of operators $L_\ell$ (also called Kraus operators), namely the parameter $m$.	
\begin{definition}
A GKLS  representation of $\mathcal{L}$ is \emph{mimimal} if the number $m$ in (\ref{eq:GKLS}) is minimal.
\end{definition}
A GKLS  representation is minimal if and only if the following condition on $V$ and $U$, that will be in force
throughout the paper, holds.

\begin{proposition}\label{prop:MinGKLS}
The  pre-generator $\mathcal{L}$ has a minimal GKLS  representation if and only if
\begin{equation} \label{eq:minimalCond}
\hbox{\rm ker}\left(V^{*}\right)\cap \hbox{\rm ker}\left(U^{\hbox{\rmseven T}}\right)=\{0\}.
\end{equation}
 \end{proposition}
\begin{proof}
A GKLS respesentation (\ref{eq:GKLS}) is minimal GKLS if and only if $\{\mathbbm{1},L_1,\ldots,L_m\}$
is a linearly independent set (see \cite{Pa}, Theorem $30.16$), namely,
$\alpha_0\mathbbm{1}+\sum_{\ell=1}^m\alpha_{\ell}L_{\ell}=0$ for $\alpha_0,\alpha_\ell\in\mathbb{C}$
implies $\alpha_{\ell}=0$ for $\ell=0,1,\ldots,m$. This identity is equivalent to
\begin{equation*}
\alpha_0\mathbbm{1}+\sum_{j=1}^d(V^*\alpha)_ja_j+\sum_{j=1}^d(U^{\hbox{\rmseven T}}\alpha)_ja_j^{\dagger}=0.
\end{equation*}
Since $\left\{\mathbbm{1}, a_1, a_1^{\dagger},\dots, a_d, a_d^{\dagger}\right\}$ is a linearly independent set,
the last equation is equivalent to $\alpha\in\hbox{\rm ker}(V^*),\alpha\in\hbox{\rm ker}(U^{\hbox{\rmseven T}}),\alpha_0=0$
and the proof is complete.
\end{proof}

\medskip
For all $x\in\mathcal{B}(\mathsf{h})$ consider the quadratic form with domain $D\times D$
\begin{eqnarray}\label{eq:Lform}
\texttt{\textrm{\pounds}}(x)  [\xi',\xi] &= & \mi\scalar{H\xi'}{x\xi} - \mi\scalar{\xi'}{x H\xi} \\
			&- & \frac{1}{2} \sum_{\ell=1}^m \left( \scalar{\xi'}{xL_\ell^*L_\ell \xi} -2\scalar{L_\ell \xi'}{xL_\ell \xi}
         + \scalar{L_\ell^* L_\ell \xi'}{x \xi} \right) \nonumber
\end{eqnarray}

We postpone to the Appendix the construction of the unique Gaussian QMS with pre-generator (\ref{eq:GKLS})
and state here the final result.

\begin{theorem}\label{th:G-QMS!}
There exists a unique QMS, $\mathcal{T}=(\mathcal{T}_t)_{t\geq 0}$ such that, for all $x\in\mathcal{B}(\mathsf{h})$ and
$\xi,\xi'\in D$, the function $t\mapsto  \scalar{\xi'}{\mathcal{T}_t (x) \xi}  $ is differentiable and
\begin{align*}
	\frac{\hbox{\rm d}}{\hbox{\rm d}t} \scalar{\xi'}{\mathcal{T}_t (x) \xi}  = \texttt{\textrm{\pounds}}(\mathcal{T}_t(x)) [\xi',\xi]
 \qquad \forall\, t\geq 0.
\end{align*}
The domain of the generator consists of $x\in\mathcal{B}(\mathsf{h})$ for which the quadratic form
$\texttt{\textrm{\pounds}}(x)$ is represented by a bounded operator.
\end{theorem}

Weyl operators do not belong to the domain of the generator of $\mathcal{T}$ because a straightforward
computation (see, for instance, Appendix A Section \ref{ssect:explWeyl}) shows that the quadratic form
$\texttt{\textrm{\pounds}}(x)$ is unbounded. In spite of this we have the following explicit formula
(see \cite{Demoen,Vheu})

\begin{theorem}\label{th:explWeyl}
Let $(\mathcal{T}_{t})_{t\ge 0}$ be the quantum Markov semigroup with generalized GKLS  generator
associated with $H,L_\ell$ as above. For all Weyl operator $W(z)$ we have
\begin{equation}\label{eq:explWeyl}
\mathcal{T}_t(W(z))
= \exp\left(-\frac{1}{2}\int_0^t \Re\scalar{\hbox{\rm e}^{sZ}z}{
  C \hbox{\rm e}^{sZ}z}\hbox{\rm d}s
+\mi\int_0^t  \Re\scalar{\zeta}{\hbox{\rm e}^{sZ}z} \hbox{\rm d}s \right)
W\left(\hbox{\rm e}^{tZ}z\right)
\end{equation}
where the \emph{real linear} operators $Z,C$ on $\mathbb{C}^d$ are
\begin{eqnarray}
  Zz &=& \left[ \left( \conj{  U^*U - V^*V}\right)/2 + \mi \Omega \right] z + \left[\left( U^TV - V^T U\right)/2 + \mi \kappa \right]\conj{z}\\
   Cz &=& \left( \conj{  U^*U + V^*V}\right) z + \left( U^TV + V^T U\right)\conj{z} \label{eq:bfC}
   \end{eqnarray}
\end{theorem}

We refer to Subsection \ref{ssect:explWeyl} for the proof.

\section{Structure of $\mathcal{N}(\T)$}\label{sect:NT}

The {\em decoherence-free} subalgebra (see \cite{AFR,CaSaUm,DFSU,DhFaRe}) of $\T$ is the defined as
\begin{equation}\label{eq-NT}
\mathcal{N}(\mathcal{T})= \left\{ x\in\mathcal{B}(\mathsf{h}) \mid
                      \mathcal{T}_t(x^*x)=\mathcal{T}_t(x^*)\mathcal{T}_t(x), \
                      \mathcal{T}_t(x x^*)=\mathcal{T}_t(x)\mathcal{T}_t(x^*), \, \forall t\ge 0
                 \right\}.
\end{equation}
This is the biggest sub von Neumann algebra of $\mathcal{B}(\mathsf{h})$ on which maps $\mathcal{T}_t$ act
as $^*$-homorphisms by the following known facts (see e.g. Evans\cite{Evans} Th. 3.1).

\begin{proposition}\label{prop:struct-NT}
Let $\mathcal{T}$ be a QMS on $\mathcal{B}(\mathsf{h})$ and let
${\mathcal{N}}(\mathcal{T})$ be the set defined by (\ref{eq-NT}). Then:
\begin{enumerate}
\item ${\mathcal{N}}(\mathcal{T})$ is $\mathcal{T}_t$-invariant for all $t\ge 0$,
\item For all $x\in {\mathcal{N}}(\mathcal{T})$ and all
$y\in\mathcal{B}(\mathsf{h})$ we have $\T_t(x^*y)=\mathcal{T}_t(x^*)\mathcal{T}_t(y)$ and
$\mathcal{T}_t(y^*x)=\mathcal{T}_t(y^*)\mathcal{T}_t(x)$,
\item $\mathcal{N}(\mathcal{T})$ is a von Neumann subalgebra of $\mathcal{B}(\mathsf{h})$.
\end{enumerate}
\end{proposition}

The decoherence-free subalgebra of a QMS with a bounded generator, i.e. written in a GKLS  form with \emph{bounded}
operators $H,L_\ell$ instead of (\ref{eq:Lell}), (\ref{eq:H}) is the commutator of the set of operators
$\delta^{n}_H(L_\ell), \delta^{n}_H(L_\ell^*)$ with $\ell=1,\dots, m,\, n\geq 0$ where $\delta_H(x)=[H,x]$.

Generators of Gaussian QMSs are represented in a generalized GKLS  form with unbounded operators $L_\ell, H$,
but $\mathcal{N}(\mathcal{T})$ can be characterized in a similar way considering \emph{generalized}
commutatant of a set of unbounded operators. We recall that, the generalized commutant of an unbounded operator
$L$ is the set of bounded operators $x$ for which $xL\subseteq Lx$, namely $Lx$ is an extension of $xL$.

We begin our investigation on $\mathcal{N}(\mathcal{T})$ by the following

\begin{theorem}\label{th:NT-comm}
The decoherence-free subalgebra of a Gaussian QMS with generator in a generalized GKLS  form associated
with operators $L_\ell,H$ as in (\ref{eq:Lell}),(\ref{eq:H}) is the generalized commutant of the following linear
combinations of creation and annihilation operators
\begin{equation}\label{eq:genNT}
\delta^{n}_H(L_\ell), \quad \delta^{n}_H(L_\ell^*) \qquad \ell=1,\dots, m,\ 0\leq n\le 2d-1
\end{equation}
where $\delta_H(x)=[H,x]$ denotes the generalized commutator and $\delta^n_H$ denotes the $n$-th iterate.
Moreover $\mathcal{T}_t(x)= \hbox{\rm e}^{\mi t H} x\, \hbox{\rm e}^{-\mi t H}$ for all $t\geq 0$
and $x\in\mathcal{N}(\mathcal{T})$.
\end{theorem}

We defer the proof to Appendix B. 
Here we give an example to show that inequality $n\le 2d-1$ is sharp.

\bigskip

\begin{example}\label{ex:2d-1}{\rm
Consider the Gaussian QMS with only one operator $L_\ell$, i.e. $m=1$ and
\[
L_1=p_1, \qquad H = q^2_d + \sum_{j=1}^{d-1} p_{j+1}q_j
\]
Compute
\[
\delta_H(L_1)= \mi p_2, \quad \delta^2_H(L_1)=-p_3,\, ...,\, \ \delta^{d-1}_H(L_1)=\mi^{d-1}p_d,
\quad \delta^{d}_H(L_1)=\mi^{d}q_d
\]
therefore $ \delta^{d+1}_H(L_1)=\mi^{d+1} q_{d-1} $.
Iterating commutators we see that $\delta^{d+k}_H(L_1)$ is proportional to $q_{d-k}$ so that, for $k=d-1$ one gets $q_1$.
Clearly, for all $k$ with $0\leq k \leq d-1$
\[
\left\{\, \delta^j_H(L_1),\, \delta^j_H(L^*_1)\,\mid\, j\le d+k \,\right\}'
\]
which is isomorphic to $L^\infty(\mathbb{R}^{d-1-k};\mathbb{C})$ (i.e. measurable functions of
momentum operators $p_1,\dots,p_{d-1+k}$).
Summing up, if we consider $2d-1$ iterated commutators we get all $p_j,q_j$ and $\mathcal{N}(\mathcal{T})$
is trivial by the irreducibility of the Weyl representation of CCR.
}
\end{example}

\bigskip

In the sequel we provide a simpler characterization of $\mathcal{N}(\mathcal{T})$ in terms of \emph{real}
subspaces of $\mathbb{C}^d$ and find its structure.
In order to make clear the thread of the discussion, we omit technicalities related with
unbounded operators that can be easily fixed because $D$ is an essential domain for
operators involved in our computations and we concentrate on the algebraic aspect.
A straightforward computation yields (with the convention of summation on repeated indexes)
\begin{eqnarray*}
[H,L_\ell] & = & \left[  \Omega_{ij}a_i^{\dagger}a_j+\frac{\kappa_{ij}}{2}a_i^{\dagger}a_j^{\dagger}
+\frac{\overline{\kappa}_{ij}}{2}a_ia_j +\frac{1}{2}\left(\zeta_{j}a_j^{\dagger}+\overline{\zeta}_ja_j\right),
\overline{v}_{\ell k} a_k + u_{\ell k} a^\dagger_k \right] \\
& = & -\left(\Omega^T\overline{v}_{\ell\bullet}\right)_ja_j + (\Omega u_{\ell\bullet})_i a^\dagger_i
-\kappa_{ij}\overline{v}_{\ell j}a^\dagger_i+\overline{\kappa}_{ij}u_{\ell j}a_i
+\frac{1}{2}\left(\langle \zeta,u_{\ell\bullet}\rangle - \langle v_{\ell\bullet},\zeta\rangle\right)\unit \\
& = &  \left(\Omega u_{\ell\bullet}-\kappa\overline{v}_{\ell\bullet}\right)_i a^\dagger_i
-\left( \Omega^T\overline{v}_{\ell\bullet} -\overline{\kappa} u_{\ell \bullet} \right)_i a_i
+\frac{1}{2}\left(\langle \zeta,u_{\ell\bullet}\rangle - \langle v_{\ell\bullet},\zeta\rangle\right)\unit.
\end{eqnarray*}
Therefore the set of operators of which we have to consider the generalized commutant, thanks to the
CCR, is particularly simple and contains only linear combinations of creation and
annihilation operators together with a multiple of the identity $\unit$ that plays no role.

Now notice that each linear combination of creation and annihilation operators is uniquely determined by a pair
$v,u$ of vectors in $\mathbb{C}^{d}$ representing coefficients of annihilation and creation operators
so that, for example, the operator $L_\ell$ in (\ref{eq:Lell}) and its adjoint $L_\ell^*$ are determined by
\[
L_\ell=\sum_{j=1}^d\left(\overline{v}_{\ell j} a_j + u_{\ell j} a^\dagger_j \right)
\rightsquigarrow \left[\begin{array}{c}
\overline{v}_{\ell\bullet} \\ u_{\ell\bullet}
\end{array}\right], \qquad
L_\ell^*=\sum_{j=1}^d\left(\overline{u}_{\ell j} a_j + v_{\ell j} a^\dagger_j \right)
\rightsquigarrow  \left[\begin{array}{c}
\overline{u}_{\ell\bullet} \\ v_{\ell\bullet}
\end{array}\right]
\]
In a similar way, after computing commutators,
\[
[H,L_\ell]\rightsquigarrow \left[\begin{array}{cc}
-\Omega^T &  \overline{\kappa}  \\
-\kappa\ & \Omega
\end{array}\right]
\left[\begin{array}{c}
\overline{v}_{\ell\bullet} \\ u_{\ell\bullet}
\end{array}\right]  , \qquad
[H,L_\ell^*]\rightsquigarrow\left[\begin{array}{cc}
-\Omega^T & \overline{\kappa} \\
-\kappa\ & \Omega
\end{array}\right]
\left[\begin{array}{c}
\overline{u}_{\ell\bullet} \\ v_{\ell\bullet}
\end{array}\right]
\]
Denote by  $\mathbbmss{H}$ the above $2d\times 2d$ matrix (built by four $d\times d$ matrices)
\[
\mathbbmss{H}  =\left[\begin{array}{cc}
-\Omega^T &  \overline{\kappa} \\
-\kappa\ & \Omega
\end{array}\right]
\]
and let $\mathcal{V}$ be the \emph{real} subspace of $\mathbb{C}^{2d}$
generated by vectors
\begin{equation}\label{eq:HnL}
\mathbbmss{H}^{\ n}\left[\begin{array}{c}
\overline{v}_{\ell\bullet} \\ u_{\ell\bullet}
\end{array}\right], \qquad
\mathbbmss{H}^{\ n}\left[\begin{array}{c}
\overline{u}_{\ell\bullet} \\ v_{\ell\bullet}
\end{array}\right]
\end{equation}
with $\ell=1,\dots, m$ and $0\leq n \leq 2d-1$.

The above remarks allow us to associate with elements of \eqref{eq:genNT} a set of vectors in $\mathbb{C}^{2d}$
and characterize the generalized commutant of \eqref{eq:genNT} in a purely algebraic way.
\begin{lemma} \label{lem:WzinNT}
An operator $x \in \mathcal{B}(\mathsf{h})$ belongs to $\mathcal{N}(\mathcal{T})$ if and only if it belongs to
the generalized commutant of
	\begin{equation} \label{eq:commQ}
		\left\{\,q(\mi w)\,\mid\, w \in \mathcal{M}\,\right\}
	\end{equation}
	 where
	\begin{equation}\label{eq:MgenNT}
	\mathcal{M} = \operatorname{Lin}_\mathbb{R}\left\{\, \mi(v+u), v-u\,\mid\, [\overline{v},u ]^{\hbox{\rmseven T}}  \in\mathcal{V}\,\right\} \subset \mathbb{C}^d.
\end{equation}
\end{lemma}
\begin{proof}
By the above remarks we know that the operators in the set \eqref{eq:genNT} are linear combination of annihilation
and creation operators up to a multiple of the identity operator and the generalized
commutant of \eqref{eq:genNT} coincides with the generalized commutant of
\begin{equation} \label{eq:commAA}
	\left\{\, a(v) + a^\dagger (u)\, \mid\, [\overline{v}, u]^{\hbox{\rmseven T}} \in \mathcal{V}\, \right\}.
\end{equation}
To conclude the proof we just need to show that the commutants of \eqref{eq:commQ} and \eqref{eq:commAA} are the same.
Notice at first that if $[\overline{v},u]^{\hbox{\rmseven T}} \in \mathcal{V}$ also $[\overline{u},v]^{\hbox{\rmseven T}} \in \mathcal{V}$,
indeed if $\delta_H^n(L_\ell) = a(v)+ a^\dagger(u)$ then
$\delta_H^n (L_\ell^*)= (-1)^n\delta_H^n(L_\ell)^* = (-1)^n \left( a(u)+ a^\dagger(v) \right)$ on the domain $D$.
Now from \eqref{eq:quadrature} we obtain
\[
\sqrt{2}\,q(\mi(u-v)) \kern-1truept = \kern-1truept \mi\kern-2truept\left(a(v)
   + a^\dagger(u) \right) -\mi  \left(a(u) + a^\dagger(v)\right), \
	\sqrt{2}(a(v)+a^\dagger(u)) \kern-1truept = \kern-1truept q(v+u) - \mi q(\mi(u-v)).
\]
Therefore every element of \eqref{eq:commQ} is a linear combination of elements of \eqref{eq:commAA} and viceversa, concluding the proof.
\end{proof}

\smallskip

In order to describe the structure of the decoherence-free subalgebra we recall now some useful definitions and properties
of symplectic spaces.
At first note that $\mathbb{C}^d$ equipped with the real scalar product $\Re\langle \cdot, \cdot \rangle$ is a real Hilbert space.
Considering instead  $\Im\langle \cdot, \cdot \rangle$ we obtain a  bilinear, antisymmetric (i.e. $\Im\langle z_1, z_2 \rangle
= - \Im\langle z_2, z_1 \rangle$, for all $z_1, z_2 \in \mathbb{C}^d$) and non-degenerate (i.e. $\Im\langle z_1, z_2 \rangle = 0$
for all $z_2 \in \mathbb{C}^d$ implies $z_1=0$) form also called a symplectic form. We now recall the following definitions.
\begin{definition}
	Let $M \subset \mathbb{C}^d$ a real linear subspace.
\begin{enumerate}
\item $M$ is a symplectic space if $\Im\langle \cdot,\cdot \rangle$ is non degenerate when restricted to elements of $M$.
\item Two elements $z_1, z_2 $ of $M$ are called symplectically orthogonal if $\Im\langle z_1,z_2 \rangle=0$.
\item Let $M_1 \subset M$ be a real linear subspace.
	We call symplectic complement of $M_1$ in $M$, and denote it by ${M_1}'$, the set
	\[
		M' = \left\{\, z \in M \,\mid\, \Im\langle z,m \rangle= 0 \quad \forall m \in M_1\, \right\}.
	\]
\item	$M_1$ is an isotropic subspace if $M_1 \subset M_1'$.
\item  $M_1$ is a symplectic subspace if $M_1$ is a symplectic space (i.e. the symplectic form $\Im\langle \cdot,\cdot \rangle$
is non degenerate when restricted to elements of $M_1$).
\end{enumerate}
\end{definition}
In order to fix some of the concepts in the above definition we provide the following
\begin{example}\label{ex:sympl}{\rm
Consider $M=\mathbb{C}^d$ which is a symplectic space and let $(e_j)_{j=1}^d$ be its canonical complex orthonormal basis. Clearly
$(e_j, \mi e_j)_{j=1}^d$ is an orthonormal basis for $\mathbb{C}^d$ considered as a real Hilbert space. Consider now a vector $e_j$
for a fixed index $j$.
It is orthogonal to all other elements of the basis with respect to the real scalar product $\Re\langle \cdot, \cdot \rangle$,
however it is symplectically orthogonal to all other elements of the basis except for $\mi e_j$. \\
	Consider now $M_1 = \operatorname{Lin}_\mathbb{R} \{ e_1, \mi e_1 \}$ and
$M_2= \operatorname{Lin}_\mathbb{R} \{ e_2 \}$, which are real linear subspaces of $\mathbb{C}^d$. It is easy to see that
	\[
		{M_1}' = \hbox{\rm Lin}_\mathbb{R} \{ e_2, \mi e_2, \dots, e_d, \mi e_d\}, \quad {M_2}'
               = \hbox{\rm Lin}_\mathbb{R} \{ e_1, \mi e_1, e_2, e_3, \mi e_3, \dots, e_d, \mi e_d\}.
	\]
	In particular $M_2 \subset M_2'$ hence it is an isotropic subspace, while $M_1$ is a symplectic subspace. \\
	Eventually it is worth noticing that not all symplectic subspaces of $M$ are also complex subspaces; from here
the need to consider real vector spaces. Indeed, consider $M_3= \operatorname{Lin}_\mathbb{R} \{ e_1 + e_2, \mi e_1\}$.
It is easy to prove that $M_3$ is a symplectic subspace of $M$ but it is not a complex subspace,
since $e_1 = (-\mi) \mi e_1 \not\in M_3$.
}
\end{example}
The previous example although seemingly simple is actually quite representative of what happens in the general case.
In analogy with classical linear algebra most complicated situations can be simplified performing a change of basis
through a homomorphism. We provide here the analogous definition for symplectic spaces.
\begin{definition}
Let $M_1, M_2 \subset \mathbb{C}^d$ be symplectic spaces. We say $B:M_1 \to M_2$ is a symplectic transformation
if it is a real linear map and moreover
	\[
		\Im\langle Bz_1, Bz_2 \rangle = \Im\langle z_1, z_2 \rangle, \quad \forall z_1,z_2 \in M_1.
	\]
	We say $B$ is a Bogoliubov transformation or symplectomorphism if it is an invertible symplectic transformation.
\end{definition}
Next Proposition collects all the properties of symplectic spaces we need (see \cite{CdS} for a comprehensive
treatment)
\begin{proposition} \label{prop:sympProp}
	Let $M \subset \mathbb{C}^d$ be a symplectic space and let $(e_j)_{j=1}^d$ be
   the canonical complex orthonormal basis of $\mathbb{C}^d$.
\begin{enumerate}
   \item There exists a symplectomorphism $B$
	\[
		B : M \to \operatorname{Lin}_\mathbb{R} \{ e_1, \mi e_1 , \dots, e_{d_1}, \mi e_{d_1} \}
	\]
 In particular $\dim_\mathbb{R} M = 2 d_1$.
	\item If $M_1 \subset M$ is a real linear subspace of $M$ then $M_1$ is a symplectic subspace if and only if $M_1 \cap {M_1}' = \{ 0 \}$. \\
	\item If $M_1$ is an isotropic subspace with $d_1 = \dim_\mathbb{R} M_1$ then there exists a symplectomorphism $B$ such that
	\[
		B : M_1 \to \operatorname{Lin}_\mathbb{R} \{ e_1, \dots, e_{d_1}\}.
	\]
\end{enumerate}
\end{proposition}
We give a proof in Appendix C for self-containedness.

\smallskip

For all subset $\mathcal{M}$ of $\mathbb{C}^d$ we denote by $\mathcal{W}(\mathcal{M})$ the von Neumann algebra generated by
Weyl operators $W(z)$ with $z\in\mathcal{M}$. H. Araki's Theorem 4 p. 1358 in \cite{Araki}, sometimes referred to as
duality for Bose fields,  (see also \cite{LeRoTe} Theorem 1.3.2 (iv) for a proof with our notation), up to a constant in
the symplectic from and also \cite{Hi} Theorem 1.1) shows that the commutant of $\mathcal{W}(\mathcal{M})$ is
$\mathcal{W}(\mathcal{M}')$.  Applying this result we can prove the following

\begin{theorem}\label{th:NT-Gauss}
The decoherence-free subalgebra $\mathcal{N}(\mathcal{T})$ is the von Neumann subalgebra of $\mathcal{B}(\mathsf{h})$
generated by Weyl operators $W(z)$ such that $z$ belongs to the sympletic complement of (\ref{eq:MgenNT}).
Moreover, up to unitary equivalence,
\begin{equation}\label{eq:structNT}
\mathcal{N}(\mathcal{T}) =  L^\infty(\mathbb{R}^{d_c};\mathbb{C})\,\overline{\otimes}\,\mathcal{B}(\Gamma(\mathbb{C}^{d_f}))
\end{equation}
for a pair of natural numbers $d_c,d_f\leq d$.
\end{theorem}

The subscript $f$ (resp. $c$) stands for full (resp. commutative).

\begin{proof}
By Lemma \ref{lem:WzinNT} any $x\in\mathcal{N}(\mathcal{T})$ satisfies
$x \, q(\mi w)\subseteq q(\mi w)\, x$ for all $w\in\mathcal{M}$. Therefore, for all real number $r$,
$x (\unit+\mi r q(\mi w))\subseteq (\unit+ \mi r q(\mi w))x$ and, right and left multiplying by the resolvent
$(\unit+\mi r q(\mi w))^{-1}$ which is a bounded operator
\[
(\unit+ \mi r q(\mi w))^{-1}x = x (\unit+\mi r q(\mi w))^{-1}.
\]
Iterating $n$ times and considering $r=1/n$ we find
$\left(\unit+ \mi q(\mi w)/n\right)^{-n}x = x \left(\unit+ \mi  q(\mi w)/n\right)^{-n}$
and, taking the limit as $n$ goes to $+\infty$, by the Hille-Yosida theorem
(\cite{BrRo} Theorem 3.1.10 p.371) we have
\[
W(w)x=\hbox{\rm e}^{-\mi q(\mi w)}x =\lim_{n\to\infty}\left(\unit+ \mi q(\mi w)/n\right)^{-n}x
= x \lim_{n\to\infty}\left(\unit+ \mi q(\mi w)/n\right)^{-n} = x W(w)
\]
and so $x$ belongs to $\mathcal{W}(\mathcal{M})'$ which coincides with $\mathcal{W}(\mathcal{M}')$
by Araki's Theorem 4 in \cite{Araki}.

Conversely, if $z$ belongs to the symplectic complement of $\mathcal{M}$, then from \eqref{eq:WzCa} and \eqref{eq:quadrature} we have $W(z)q(\mi w)e_g=q(\mi w)W(z)e_g$ for
all $w\in\mathcal{M}$ and $g\in\mathbb{C}^d$. Since the linear span of exponential vectors is an essential domain
for $q(\mi w)$, for all $\xi\in\hbox{\rm Dom}(q(\mi w))$ there exists a sequence $(\xi_n)_{n\geq 1}$ in $\mathcal{E}$
such that $(q(\mi w)\xi_n)_{n\geq 1}$ converges to $q(\mi w)\xi$.
It follows that $(q(\mi w)W(z)\xi_n)_{n\geq 1}$ converges and, since $q(\mi w)$ is closed, $W(z)\xi$ belongs to
$\hbox{\rm Dom}(q(\mi w))$ and $W(z)q(\mi w)\xi=q(\mi w)W(z)\xi$,
namely $q(\mi w)W(z)$ is an extension of $W(z)q(\mi w)$.
Therefore $W(z)$ belongs to the generalized commutant of all $q(\mi w)$ with $w\in \mathcal{M}$ and therefore to $\mathcal{N}(\mathcal{T})$ by Lemma \ref{lem:WzinNT}.

In order to prove (\ref{eq:structNT}) consider $\mathcal{M}^c:= \mathcal{M} \cap \mathcal{M}'$ which is a real linear subspace of both $\mathcal{M}$ and $\mathcal{M}'$. Consider now $\mathcal{M}^r$ and $\mathcal{M}^f$ as the real linear complement of $\mathcal{M}^c$ in $\mathcal{M}$ and $\mathcal{M}'$ respectively, i.e.
\[
	\mathcal{M} = \mathcal{M}^c \oplus \mathcal{M}^r, \qquad \mathcal{M}' = \mathcal{M}^c \oplus \mathcal{M}^f.
\]
($\mathcal{M}^c \perp \mathcal{M}^r$ and $\mathcal{M}^c \perp \mathcal{M}^f$, more precisely, they are orthogonal with respect to
the real part of the scalar product).
We will show that $\mathcal{M}^f$ is a symplectic subspace of $\mathbb{C}^d$ and that it is symplectically orthogonal to both
$\mathcal{M}^c$ and $\mathcal{M}^r$.
Suppose $z \in \mathcal{M}^f$ is such that $\Im\langle z, z_f \rangle=0$ for all $z_f \in \mathcal{M}^f$.
By construction $\Im\langle z, z_c \rangle=0$ for all $z_c \in \mathcal{M}^c=\mathcal{M} \cap \mathcal{M}'$.
Therefore
\[
	\Im\langle z, m \rangle = 0, \quad \forall m \in \mathcal{M}',
\]
since $\mathcal{M}' = \mathcal{M}^c \oplus \mathcal{M}^f$. Therefore $z \in \mathcal{M}''=\mathcal{M}$, but then
\[
 z \in \mathcal{M} \cap \mathcal{M}' \cap \mathcal{M}^f = \mathcal{M}^c \cap \mathcal{M}^f = \{0 \}.
\]
Hence $\mathcal{M}^f$ is a symplectic subspace. Eventually, $\mathcal{M}^f \subset \mathcal{M}'$ and $\mathcal{M}=\mathcal{M}'' \subset (\mathcal{M}^f)'$. In particular $\mathcal{M}^f$ is symplectically orthogonal to both $\mathcal{M}^r, \mathcal{M}^c$.
Let $d_c = \dim_\mathbb{R} \mathcal{M}^c$ and $2d_f= \dim_\mathbb{R} \mathcal{M}^f$ which is even by Proposition \ref{prop:sympProp} 1.
Still by Proposition \ref{prop:sympProp} we can find a symplectomorphism $B$ such that
\[
	 B : \mathcal{M}' \to \operatorname{Lin}_\mathbb{R} \{ e_1, \dots, e_{d_c} \} \oplus \operatorname{Lin}_\mathbb{R}\{ e_{d_c+1}, \mi e_{d_c+1}, \dots, e_{d_c+d_f}, \mi e_{d_c+d_f}\},
\]
where $(e_j)_{j=1}^{d_c+d_f}$ is the canonical complex orthonormal basis of $\mathbb{C}^{d_c+d_f}$. Eventually, since symplectic
transformation in finite dimensional symplectic spaces are always implemented by unitary transformations on the Fock space
(see \cite{Derez} Theorem 3.8), we obtain the final result.
\end{proof}

\noindent{\bf Remark.} An analogous argument to the proof of the previous theorem allows us to show that also
$\mathcal{M}^r$ is a symplectic subspace which is symplectically orthogonal to both $\mathcal{M}^c$ and $\mathcal{M}^f$.
If  $2d_r=\dim_\mathbb{R} \mathcal{M}^r$, in total analogy with the proof, we can always find a symplectomorphism such that
\[
	\mathcal{M} \cup \mathcal{M}' \simeq \mathbb{C}^{d_r} \oplus \mathcal{M}^c \oplus \mathbb{C}^{d_f},
\]
where $\mathcal{M}_c$ is the real subspace of $\mathbb{C}^{d_c}$ generated by $\{e_1, \dots, e_{d_c}\}$. In particular,
after the unitary transformation associated with the symplectic transformation, we have
\[
	\mathcal{W}(\mathcal{M} \cup \mathcal{M}') = \mathcal{B}(\Gamma(\mathbb{C}^{d_r})) \overline{\otimes} L^\infty (\mathbb{R}^{d_c} ; \mathbb{C}) \overline{\otimes} \mathcal{B}(\Gamma(\mathbb{C}^{d_f})).
\]

\noindent{\bf Remark.} It is worth noticing here that a QMS with $\mathcal{N}(\mathcal{T})$
as in (\ref{eq:structNT}) does not necessarily admit a dilation with $d_c$ classical noises because
the corresponding Kraus operators $L_\ell$ could be normal but not self-adjoint (see subsection \ref{ssect:1L-H=0}
for an example) and so one may find obstructions to dilations with classical noises as shown in
\cite{FGNV}.

\begin{corollary}\label{cor:NT-Z}
The decoherence-free subalgebra $\mathcal{N}(\mathcal{T})$ is generated by Weyl operators $W(z)$ with
$z$ belonging to real subspaces of $\hbox{\rm ker}(C)$ that are $Z$-invariant.
\end{corollary}

\begin{proof}
By Theorem \ref{th:NT-Gauss} it suffices to show that $z$ belongs to the symplectic complement $\mathcal{M}'$ of (\ref{eq:MgenNT})
if and only if it belongs to a real subspace of $\hbox{\rm ker}(C)$ that is $Z$-invariant.\\
If $z$ belongs to $\mathcal{M}'$ then $W(z)\in \mathcal{N}(\mathcal{T})$ and
$\mathcal{T}_t(W(z))= \hbox{\rm e}^{\mi t H} W(z)\, \hbox{\rm e}^{-\mi t H}$ for all $t\geq 0$.
Comparison with (\ref{eq:explWeyl}) yields
\[
\hbox{\rm e}^{\mi t H} W(z)\, \hbox{\rm e}^{-\mi t H} =
\exp\left(-\frac{1}{2}\int_0^t \Re\left\langle\hbox{\rm e}^{sZ}z,
  C \hbox{\rm e}^{sZ}z\right\rangle\hbox{\rm d}s
+\mi\int_0^t  \Re\left\langle\zeta, \hbox{\rm e}^{sZ}z\right\rangle \hbox{\rm d}s \right)
W\left(\hbox{\rm e}^{tZ}z\right)
\]
Unitarity of both left and right operators implies
$\Re\left\langle\hbox{\rm e}^{sZ}z, C \hbox{\rm e}^{sZ}z\right\rangle=0$ for all $s\geq 0$
and $\hbox{\rm e}^{sZ}z$ belongs to $\hbox{\rm ker}(C)$ for all $s\geq 0$, namely, in an equivalent way,
$z$ and also $Zz$ (by differentiation) belong to $\hbox{\rm ker}(C)$. \\
Conversely, if $z$ belongs to a real subspace of $\hbox{\rm ker}(C)$ that is $Z$-invariant,
then $\hbox{\rm e}^{sZ}z$ also belongs to that subset for all $s\geq 0$.
The explicit formula (\ref{eq:explWeyl}) shows that
\[
\mathcal{T}_t(W(z))=\exp\left(\mi\int_0^t  \Re\left\langle\zeta, \hbox{\rm e}^{sZ}z\right\rangle \hbox{\rm d}s \right)
W\left(\hbox{\rm e}^{tZ}z\right)
\]
therefore
\[
\mathcal{T}_t(W(z)^*)\mathcal{T}_t(W(z))= \hbox{\rm e}^{\mi t H} W(z)^*W(z)\, \hbox{\rm e}^{-\mi t H}
=\unit =\mathcal{T}(W(z)^*W(z))
\]
and, in the same way, $\mathcal{T}_t(W(z))\mathcal{T}_t(W(z)^*)=\mathcal{T}(W(z)W(z)^*)$. It follows that
$W(z)\in \mathcal{N}(\mathcal{T})$ and $z$ belongs to the symplectic complement of (\ref{eq:MgenNT}) by Theorem \ref{th:NT-Gauss}.
\end{proof}

The following corollary shows that we can perform a unitary transformation of the Fock space in order to reduce the number of creation and annihilation operators that appear in the Kraus' operators.
\begin{corollary} \label{cor:KrausMultReduct}
There exists a unitary transformation $U$ of the Fock space such that
\[
U L_\ell U^* = \sum_{j=1}^{d_r+d_c} (\overline{v}_j a_j + u_j a^\dagger_j)
\]
\end{corollary}

\noindent{\bf Proof.} It suffices to consider the transformation obtained in the Remark after
Theorem \ref{th:NT-Gauss}. Indeed each Kraus operator corresponds to a vector $[\overline{v}, u]^{\hbox{\scriptsize T}} \in \mathcal{V}$
which in turn corresponds to two generators in the subspace $\mathcal{M}$. Performing the symplectomorphism in the cited Remark we have
\[
	\mathcal{M} \simeq \mathbb{C}^{d_r} \oplus \mathcal{M}^c
\]
which has dimension $2d_r + d_c$. In particular if $U$ is the unitary transformation that implements this symplectomorphism
$U L_\ell U^*$ will depend at most from $d_r+ d_c$ modes. \hfill $\square$

\smallskip

{\bf Example.} One may wonder if $H$ can also be written in a special form in the new representation
of the CCR, for example as the sum of two self-adjoint operators, one depending only on $b_{1},b^\dagger_{1},\dots,b_{d_r+d_c},b^\dagger_{d_r+d_c}$ and the other depending only
on $b_{d_r+d_c+1},b^\dagger_{d_r+d_c+1},\dots,b_{d},b^\dagger_{d}$.  This happens when
$\mathcal{N}(\mathcal{T})$ is a countable sum of type I factors (see \cite{DFSU}) but not in the
case of Gaussian QMSs with $d_c>0$ as shows this example. \\
Let $d=2$, $m=1$ and
\[
L = q_1, \qquad H = q_1 p_2.
\]
Clearly, by Theorem \ref{th:NT-comm}, $\mathcal{N}(\mathcal{T})$ is the algebra
$L^\infty(\mathbb{R};\mathbb{C})\overline{\otimes}\mathcal{B}(\Gamma(\mathbb{C}))$ but
$H$ is the product of two operators depending on different coordinates.

\section{Applications}\label{sect:appl}

In this section we present two examples to illustrate the admissible structures of decoherence-free subalgebras
of a Gaussian QMS on $\mathcal{B}(\Gamma(\mathbb{C}^d))$ with $d\geq 2$ and the application to an open system
of two bosons in a common environment (see Ref. \cite{CGMZ}). We begin by considering the case of only one noise operator.

\subsection{The case one $L$, $H=N$}\label{ssect:1L-H=N}
The operators (\ref{eq:Lell}) and $(\ref{eq:H})$ are the closure of operators defined on $D$
\begin{equation}\label{eq:LN-QMS}
L=\sum_{j=1}^d\left( \overline{v}_j a_j + u_j a^\dagger_j\right), \qquad H=\sum_{j=1}^d a_j^\dagger a_j
\end{equation}
(either $v$ or $u$ is nonzero). We compute recursively
\[
\delta^{2n+1}_H(L) = \sum_{j=1}^d\left(u_j a^\dagger_j-\overline{v}_j a_j\right), \qquad
\delta^{2n}_H(L) = L
\]
for all $n\geq 0$, and, in the same way, $\delta^{2n+1}_H(L^*) = \sum_{j=1}^d\left(v_j a^\dagger_j-\overline{u}_j a_j\right)$,
$\delta^{2n}_H(L^*) = L^*$.
It follows that $\mathcal{M}'$ is the symplectic complement of
\[
\hbox{\rm Lin}_{\mathbb{R}}\left\{\, v-u, v+u, \mi(v+u),\mi (v-u)\,\right\}
=\hbox{\rm Lin}_{\mathbb{R}}\left\{\, v, u, \mi v ,\mi u\,\right\}
=\hbox{\rm Lin}_{\mathbb{C}}\left\{\, v, u \,\right\}.
\]
Thus $\mathcal{M}'$ is the orthogonal (for the complex scalar product)
of the complex linear subspace generated by $v$ and $u$, it is a complex subspace
of $\mathbb{C}^d$ and
\[
\mathcal{N}(\mathcal{T}) = \left\{\, W(z)\,\mid\, z\in\mathcal{M}'\,\right\}=\mathcal{B}(\Gamma(\mathcal{M}')).
\]
If $v,u$ are linearly independent, then the complex dimension of $\mathcal{M}'$ is $d-2$,
and $\mathcal{N}(\mathcal{T})$ is isomorphic to $\mathcal{B}(\Gamma(\mathbb{C}^{d-2}))$.

\subsection{The case one $L$, $H=0$}\label{ssect:1L-H=0}
Let $L$ be as in (\ref{eq:LN-QMS}). If $H=0$, then $\delta_H =0$. In particular
\[
	\mathcal{M} = \hbox{\rm Lin}_{\mathbb{R}} \{ v-u,u-v, \mi (v+u), \mi (u+v) \} = \hbox{\rm Lin}_{\mathbb{R}} \{ v-u, \mi(v+u) \}
\]
and, since both $v$ and $u$ cannot be zero in our framework, $\dim_\mathbb{R} \mathcal{M}$ is either $1$ or $2$.
If it is equal to $1$ (first case), clearly $\mathcal{M}\cap\mathcal{M}'=\mathcal{M}$ and $\mathcal{M}^r=\{0\}$
therefore $d_c=1$ and $d_r=0$. It follows that $d_f=d-1$ and $\mathcal{N}(\mathcal{T})$ is a von Neumann algebra
unitarily equivalent to $L^\infty(\mathbb{R}; \mathbb{C} ) \overline{\otimes} \mathcal{B}(\Gamma(\mathbb{C}^{d-1}))$.
If $\dim_\mathbb{R} \mathcal{M}=2$ then, since $\mathcal{M}^r$ is a symplectic space, its real dimension must be even and
so we distinguish two cases: $d_r=0,d_c=2$ (second case) and $d_r=1, d_c=0$ (third case). If $d_c=2$, again
$\mathcal{M}\cap\mathcal{M}'=\mathcal{M}$, and $\mathcal{N}(\mathcal{T})$ is a von Neumann algebra
$L^\infty(\mathbb{R}^2; \mathbb{C} ) \overline{\otimes} \mathcal{B}(\Gamma(\mathbb{C}^{d-2}))$.  \\
If $d_c=0,d_r=1$, then $d_f=d-1$ and $\mathcal{N}(\mathcal{T})$ is a von Neumann algebra
$\mathcal{B}(\Gamma(\mathbb{C}^{d-1}))$. This classification is summarized by Table \ref{tab:oneL0}
in which the last column labeled ``$L$'' contains possible choices of the operator $L$ that realize each case.
\begin{table}[ht]
\centering
\begin{tabular}{|c|c|c|c|c|c|c|}
  \hline
  Case   & $\dim_\mathbb{R} \mathcal{M}$ & $d_c$ & $d_r$ & $d_f$ & $\mathcal{N}(\mathcal{T})$ &  $L$ \\
  \hline
  1$^{\hbox{\rmseven st}}$ & 1 & 1 & 0 & $d-1$ & $L^\infty(\mathbb{R}; \mathbb{C} ) \overline{\otimes} \mathcal{B}(\Gamma(\mathbb{C}^{d-1}))$
  & $L=q_1$  \\
  \hline
  2$^{\hbox{\rmseven nd}}$  & 2 & 2 & 0 & $d-2$ & $L^\infty(\mathbb{R}^2;\mathbb{C})\overline{\otimes} \mathcal{B}(\Gamma(\mathbb{C}^{d-2}))$
  & $L=q_1 + \mi q_2$ \\
  \hline
  3$^{\hbox{\rmseven rd}}$  & 2 & 0 & 1 & $d-1$ & $\mathcal{B}(\Gamma(\mathbb{C}^{d-1}))$\ & $L= a_1$ or $L=a^\dagger_1$ \\
  \hline
\end{tabular}
\caption{ $\mathcal{N}(\mathcal{T})$ that can arise with one $L$ and $H=0$ }
\label{tab:oneL0}
\end{table}

In the last part of the section we will characterize each case by just looking directly at the operator $L$ instead of computing $\mathcal{M}$.

Suppose $L$ is self-adjoint. In this case $\mathcal{V}$ is composed of only one vector which is of the form
$[\overline{v},v]^{\hbox{\rmseven T}}$. Therefore $\mathcal{M}= \hbox{\rm Lin}_\mathbb{R} \{ \mi v\}$ and $d_c=1$, while $d_r=0$
(1$^{\hbox{\rmseven st}}$ case). \\
Consider now instead the case $L$ normal but not self-adjoint. An explicit computation shows that
$	0 = [L, L^*] = \| v \|^2 - \| u \|^2$ on $D$.
This condition shows that $\mathcal{M} = \mathcal{M}\cap \mathcal{M}'$ since
\[
	\Im \langle v-u,\mi (v+u) \rangle= \| v \|^2 - \| u \|^2=0.
\]
Moreover $u \neq v$ since $L$ is not self-adjoint, hence $d_c=2$ (2$^{\hbox{\rmseven nd}}$ case).
If $L$ is not even normal (i.e. $\| v \|^2 \neq \| u \|^2$) then by the previous calculations $d_c=0$ and $d_r=1$
(3$^{\hbox{\rmseven rd}}$ case).\\
Summing up: the 1$^{\hbox{\rmseven st}}$ case arises when $L$ is self-adjoint, the case 2$^{\hbox{\rmseven nd}}$ case arises when $L$ is normal but not self-adjoint and the 3$^{\hbox{\rmseven rd}}$ case arises when $L$ is not normal or, equivalently $\| v \|^2 \neq \| u \|^2$. \\
In the last case it can be shown that when $\| v \|^2 > \| u \|^2$ (resp. $\| v \|^2 < \| u \|^2$) there
exists a Bogoliubov transformation changing $L$ to a multiple of the annihilation operator $a_1$ (resp. creation operator $a^\dagger_1$.

\subsection{Two bosons in a common bath}

The following model for the open quantum system of two bosons in a common environment has been considered in
Ref. \cite{CGMZ}. Here $d=2$ and $H$ is as in equation \ref{eq:H} with $\kappa=\zeta=0$.
The completely positive part of the GKLS generator $\mathcal{L}$ is
\begin{equation}\label{eq:GKLS-CGMZ-cp}
\frac{1}{2}\sum_{j,k=1,2} \gamma^{-}_{jk}\, a^\dagger_j X a_k
+ \frac{1}{2}\sum_{j,k=1,2} \gamma^{+}_{jk}\, a_j X a^\dagger_k
\end{equation}
where $(\gamma^{\pm}_{jk})_{j,k = 1,2}$ are positive definite $2\times 2$ matrices.

Note that, by a change of phase $a_1\to \hbox{\rm e}^{\mi\theta_1} a_1$,
$a^\dagger_1\to\hbox{\rm e}^{-\mi\theta_1} a^\dagger_1$,
$a_2\to \hbox{\rm e}^{\mi\theta_2} a_2$, $a^\dagger_2\to \hbox{\rm e}^{-\mi\theta_2} a^\dagger_2$,
we can always assume that $(\gamma^{-}_{jk})_{j,k = 1,2}$ is \emph{real} symmetric.
Write the spectral decomposition
\[
\gamma^{\pm} = \lambda_{\pm}|\varphi^{\pm}\rangle\langle \varphi^{\pm}| + \mu_{\pm}|\psi^{\pm}\rangle\langle \psi^{\pm}|
\]
where the vectors $\varphi^{-},\psi^{-}$ have \emph{real} components.
Rewrite the first term of (\ref{eq:GKLS-CGMZ-cp}) as
\begin{eqnarray*}
\sum_{j,k=1,2} \gamma^{-}_{jk}\, a^\dagger_j X a_k
    &=& \lambda_{-}\sum_{j,k=1,2} \varphi^{-}_j \varphi^{-}_k \, a^\dagger_j X a_k
    + \mu_{-}\sum_{j,k=1,2} \psi^{-}_j \psi^{-}_k \, a^\dagger_j X a_k\\
    &=& \lambda_{-}\left(\sum_{j=1,2} \varphi^{-}_j \, a^\dagger_j\right) X
    \left(\sum_{k=1,2}\varphi^{-}_k a_k\right)
    + \mu_{-}\left(\sum_{j=1,2} \psi^{-}_j \, a^\dagger_j\right) X
    \left(\sum_{k=1,2}\psi^{-}_k a_k\right)
\end{eqnarray*}
and write in a similar way the second term of (\ref{eq:GKLS-CGMZ-cp})
\[
\sum_{j,k=1,2} \gamma^{+}_{jk}\, a_j X a^\dagger_k
    = \lambda_{+}\left(\sum_{j=1,2} \overline{\varphi^{+}}_j \, a_j\right) X
    \left(\sum_{k=1,2}\varphi^{+}_k a^\dagger_k\right)
    + \mu_{+}\left(\sum_{j=1,2} \overline{\psi^{+}}_j \, a_j\right) X
    \left(\sum_{k=1,2}\psi^{+}_k a^\dagger_k\right)
\]
We can represent $\mathcal{L}$ in a generalized GKLS form with a number of Kraus operators $L_\ell$ depending
on the number of strictly positive eigenvalues among $\lambda_{\pm},\mu_{\pm}$.
\begin{eqnarray*}
L_1 = \lambda_{-}^{1/2} \sum_{k=1,2} \varphi_k^{-} a_k & \qquad &
L_2 = \mu_{-}^{1/2} \sum_{k=1,2} \psi_k^{-} a_k \\
L_3 = \lambda_{+}^{1/2} \sum_{k=1,2}\varphi_k^{+} a^\dagger_k & \qquad  &
L_4 = \mu_{+}^{1/2} \sum_{k=1,2}\psi_k^{+} a^\dagger_k
\end{eqnarray*}
Relabelling if necessary, we can always assume $0\leq\lambda_{-} \leq \mu_{-}$ and
$0\leq\lambda_{+} \leq \mu_{+}$.

We begin our analysis by considering  the case where $H=0$.

If $\lambda_{-}>0$ (or $\lambda_{+}>0$) then there are four vectors $v,u$ in the defining set of
$\mathcal{M}$ namely
\[
\mathcal{M} = \hbox{\rm Lin}_{\mathbb{R}}\left\{\, \varphi^{-}, \psi^{-}, \mi \varphi^{-}, \mi\psi^{-} \,\right\}
\quad (\hbox{\rm or\ }  = \hbox{\rm Lin}_{\mathbb{R}}\left\{\, \varphi^{+}, \psi^{+}, \mi \varphi^{+}, \mi\psi^{+} \,\right\})
\]
thus $\mathcal{M}'=\{0\}$ and $\mathcal{N}(\mathcal{T}) = \mathbb{C}\unit$.

Suppose now that $\lambda_{+}=\lambda_{-}=0$ and $\mu_{-},\mu_{+}>0$ so that there are only two Kraus operators,
the above $L_2$ and $L_4$ and
\[
\mathcal{M} = \hbox{\rm Lin}_{\mathbb{R}}\left\{\,\psi^{-}, \psi^{+}, \mi\psi^{-}, \mi\psi^{+} \,\right\}.
\]
It follows that, if $\psi^{-}, \psi^{+}$ are $\mathbb{R}$-linearly independent, we have again
$\mathcal{M}=\mathbb{C}^2$ whence $\mathcal{M}'=\{0\}$ and $\mathcal{N}(\mathcal{T}) = \mathbb{C}\unit$.
Otherwise, if $\psi^{+}$ is a \emph{real} non-zero multiple of $\psi^{-}$, then, as $\psi^{\pm}$ and $\mi\psi^{\pm}$
are $\mathbb{R}$-linearly independent, the real dimension of $\mathcal{M}$ and $\mathcal{M}'$ is two,
$\mathcal{M}\cap\mathcal{M}'=\{0\}$ so that $\mathcal{N}(\mathcal{T})$ is isomorphic to $\mathcal{B}(\Gamma(\mathbb{C}))$.

It is not difficult to see that, in any case, the dimension of $\mathcal{M}$
cannot be $1$ or $3$ (because creation and annihilation operator always appear separately in different
Kraus operators $L$, never in the same).

Summarizing: $\mathcal{N}(\mathcal{T})$ is non-trivial and isomorphic to $\mathcal{B}(\Gamma(\mathbb{C}))$
if and only if $\gamma^{+}$ and $\gamma^{-}$ are rank-one and commute.

Finally, if we consider a non-zero $H$, it is clear that $\mathcal{N}(\mathcal{T})$ is always trivial
unless $\gamma^{+}$ and $\gamma^{-}$ are rank-one, commute and their one-dimensional range is an
eigenvector for $\Omega$ and $\Omega^T$.

\section*{Appendix A: construction of Gaussian QMSs from the GKLS generator}

In this section we outline how one can construct the minimal quantum dynamical semigroup associated with operators
$H, L_\ell$ and following \cite{Fa}, Section 3.3. The first step is to prove that the closure of the operator $G$
\begin{equation} \label{eq:G}
	G = -\frac{1}{2} \sum_{\ell=1}^m L_\ell^* L_\ell -\mi H
\end{equation}
defined on the domain $D$ generates a strongly continuous contraction semigroup.
To this end we recall the result due to  Palle E.T. Jorgensen (see \cite{Jo}, Theorem $2$).

\begin{theorem}\label{th:Gdis}
Let $G$ be a dissipative linear operator on a Hilbert space $\mathsf{h}.$ Let $(D_n)_{n\geq1}$ be
an increasing family of closed subspaces of $\mathsf{h}$ whose union is dense in $\mathsf{h}$ and contained in
the domain of $G$ and let $P_{D_n}$ be the orthogonal projection of $\mathsf{h}$ onto $D_n.$ Suppose that there
exists an integer $n_0$ such that $GD_n\subset D_{n+n_0}$ for all $n\geq 1$. Then the closure $\overline{G}$ generates a
strongly continuous contraction semigroup on $\mathsf{h}$ and $\cup_{n\geq1}D_n$ is a core for $\overline{G}$, if there exists
a sequence $(c_n)_{n\geq1}$ in $\mathbb{R}_+$ such that $
|| GP_{D_n} - P_{D_n}GP_{D_n}||\leq c_n$  for all $n$ and
$$\sum\limits_{n=1}^{\infty}c_n^{-1}=\infty$$
\end{theorem}

We are now able to prove the following proposition.

\begin{proposition} \label{prop:Ggenerates}
The operator $\overline{G}$ is the infinitesimal generator of a strongly continuous contraction semigroup on $\mathsf{h}$
and $D$ is a core for this operator.
\end{proposition}
\begin{proof}
We apply Theorem \ref{th:Gdis} with $D_n$ the linear manifold spanned by vectors $e(n_1,\ldots,n_d)$ with $n_1+\ldots+n_d\leq n$. Clearly $D=\cup_{n\geq1}D_n$. The operator $G$ is obviously densely defined and dissipative. Therefore it is closable (see e.g. \cite{BrRo}, Lemma 3.1.14) and its closure, denoted $\overline{G}$ is dissipative.
Clearly, by the explicit form of the action of creation and annihilation operators on
vectors $e(n_1,\ldots,n_d)$, the operator $G$ maps $D_n$ into $D_{n+2}$ for all $n\geq 0$.

A straightforward computation using $(\ref{eq:G})$  yields
\[
	(GP_{D_n} - P_{D_n}GP_{D_n}) = \left(-\frac{1}{2}\sum_{k,j=1}^d
	\left( \sum_{\ell=1}^m v_{\ell	k} u_{\ell j} +\mi \kappa\right)a^\dagger_k a^\dagger_j
   -\frac{\mi}{2} \sum_{k=1}^d \zeta_k a^\dagger_k\right)P_{D_n},
\]
namely the non-zero part is the one involving only creations.
Let us fix $u=\sum_{|\alpha|\leq n}r_{\alpha}e_{\alpha}$ a vector in $D_n$, where $\alpha=(\alpha(1),\ldots,\alpha(d))$ is a multi-index, $|\alpha|=\alpha(1)+\ldots +\alpha(d)$, and the vector
$e_{\alpha}^{\hbox{\rmseven T}}=(e_{\alpha(1)},\ldots,e_{\alpha(d)})$. Clearly $a^\dagger_j u \in D_{n+1}$ and
\[
	\norm{ a^\dagger_j u}^2 \leq \sum_{|\alpha|\leq n} \modulo{r_{\alpha}}^2 \left( \alpha (j) +1 \right) \norm{e_{\alpha+1_j}}^2 \leq (n+1) \norm{u}^2.
\]
Therefore we have also
\[
	\norm{a^\dagger_j a^\dagger_k u}^2 \leq (n+2) \norm{a^\dagger_k u}^2 \leq (n+2)(n+1)\norm{u}^2 \leq (n+2)^2\norm{u}^2.
\]
This means that
\begin{align*}
	\norm{(GP_{D_n} - P_{D_n}GP_{D_n}) u}
    & \leq \frac{1}{2}\sum_{j,k=1}^d \norm{ \left[  \sum_\ell \left(\overline{v}_{\ell	\bullet}\right)^* u_{\ell \cdot}
    +\mi \kappa \right]_{jk} a^\dagger_j a^\dagger_ku} + \sum_{j=1}^d \norm{\zeta_j a^\dagger_{j}u}  \\
	&\leq c(n+2) \norm{u}/2
\end{align*}
with $c>0$ a constant that does not depend on $n$. Since the series $\sum_{n\geq 1} (n+2)^{-1}$ diverges we can apply
Theorem \ref{th:Gdis} and the proposition is proved.
\end{proof}

A similar argument  allows us to prove the following

\begin{proposition}\label{prop:Phi-s.a.}
The closure $\Phi$ of the operator $\sum_{1\leq\ell\leq m}L_\ell^*L_\ell$ defined on the domain $D$
is essentially self-adjoint.
\end{proposition}

Let us denote by $G$ also the closure of the operator $G$ to simplify the notation.
Using standard arguments, the operators $L_{\ell}$ can be extended to the domain of $G$ and further extended to the domain
$
\left\{\,u\in\mathsf{h}\,\mid\, u=\sum_{\alpha}r_{\alpha}e_{\alpha},\sum_{\alpha}|\alpha| \,|r_{\alpha}|^2<\infty\,\right\},
$
where $\alpha=(\alpha(1),\ldots,\alpha(d))$ is a multi-index, $|\alpha|=\alpha(1)+\ldots +\alpha(d)$, and the vector $e_{\alpha}=(e_{\alpha(1)},\ldots,e_{\alpha(d)})$.

In the next section we will show that the minimal semigroup is identity preserving and so it is a well defined QMS, whose predual
semigroup is trace preserving.

\subsection{Conservativity}

We will establish conservativity by applying the Chebotarev-Fagnola
sufficient condition (see \cite{ChFa} , \cite{Fa} section 3.5). More precisely, we will apply the following result:

\begin{theorem}\label{th:cons}
Suppose that:
\begin{enumerate}
\item the operator $G$ is the infinitesimal generator of a strongly continuous contraction semigroup $(P_t)_{t\geq0}$ in $\mathsf{h}$,
\item the domain of the operators $(L_\ell)_{\ell\geq 1}$
contains the domain of $G$ and, for every $u\in D(G)$, we have
\[\langle u,Gu\rangle+\langle Gu,u\rangle+\sum\limits_\ell\langle L_\ell u,L_\ell u\rangle=0,\]
\item there exists a self-adjoint operator
$C$ with domain coinciding with the domain of $G$ and a core $D$ for $C$ with the following properties
\begin{itemize}
\item[(a)] $L_\ell(D)\subset D(C^{1/2})$ for all $\ell\geq1$,
\item[(b)] there exists a self-adjoint operator $\Phi$ such that
\[-2\Re \langle u,Gu\rangle=\langle u,\Phi u\rangle\leq\langle u,C u\rangle\]
for all $u\in D$,
\item[(c)] there exists a positive constant $b$ such that the inequality
\[2\Re \langle Cu,Gu\rangle+\sum_{\ell\geq 1}\langle C^{1/2}L_\ell u,C^{1/2}L_\ell u\rangle\leq b\langle u,C u\rangle\]
holds for every $u\in D.$
\end{itemize}
\end{enumerate}
Then the minimal quantum dynamical semigroup is Markov.
\end{theorem}

In order to check the above conditions one should proceed with
computations on quadratic forms. However, these are equivalent
to algebraic computations of the action of the formal generator
$ \pounds $ on first and second order polynomials is $a_j,a^\dagger_j$ therefore we will go on with algebraic computations
so as to reduce the clutter of the notation.

\begin{lemma}\label{lemmaak}
It holds
\begin{eqnarray*}
\pounds (a_k)  &=&  - \frac{\mathrm{i}\zeta_k}{2}\mathbbm{1} + \frac{1}{2}\sum_{j=1}^d \left(
(U^{\hbox{\rmseven T}}V-V^{\hbox{\rmseven T}}U-2\mathrm{i}\kappa)_{kj}a^\dagger_j
+(U^{\hbox{\rmseven T}}\overline{U}-V^{\hbox{\rmseven T}}\overline{V}-2\mathrm{i}\Omega)_{kj}a_j
\right), \\ 
\pounds (a_k^\dagger) &=&  \frac{\mathrm{i}\overline{\zeta_k}}{2}\mathbbm{1}
+ \frac{1}{2}\sum_{j=1}^d \left( (U^*\overline{V}-V^*U+2\mathrm{i}\kappa^*)_{kj}a_j
+(U^*U-V^*V+2\mathrm{i}\overline{\Omega})_{kj}a_j^\dagger\right). 
\end{eqnarray*}
\end{lemma}
\begin{proof}
First write
\[
\pounds_0(X):=-\frac{1}{2}\sum_{\ell=1}^m\left(L_\ell^*L_\ell X-2L_\ell^*XL_\ell+XL_\ell^*L_\ell\right)
=\frac{1}{2}\sum_{\ell=1}^m\left(L_\ell^*[X,L_\ell] +[L_\ell^*,X]L_\ell\right).
\]
By the CCR one has
\begin{equation*}
[a_k,L_\ell] =\sum_{j=1}^d [a_k,u_{\ell j}a_j^{\dagger}]=u_{\ell k}\mathbbm{1}, \quad
[L_\ell^*,a_k]=[a_k^{\dagger},L_\ell]^* =\sum_{j=1}^d [a_k^{\dagger},\overline{v}_{\ell j}a_j]^*=-v_{\ell k}\mathbbm{1}
\end{equation*}
Therefore we obtain that
\begin{equation*} 
\pounds_0(a_k)=\frac{1}{2}\sum_{\ell=1}^m\left\{\left(\sum_{j=1}^d v_{\ell j}a_j^{\dagger}+\overline{u}_{\ell j}a_j\right)u_{\ell k}-v_{\ell k}\left(\sum_{j=1}^d\overline{v}_{\ell j}a_j+u_{\ell j}a_j^{\dagger}\right)\right\}.
\end{equation*}
and
\begin{equation*} 
[H,a_k]=-\frac{\mi\zeta_k}{2}\mathbbm{1}-\sum_{j=1}^d\left(\Omega_{kj}a_j+\kappa_{kj} a_j^{\dagger}\right).
\end{equation*}
which both lead to
\begin{eqnarray*}
\pounds (a_k) & = &\mathrm{i} [H,a_k]+\pounds_0 (a_k)\\
& = &-\frac{\zeta_k}{2}\mathbbm{1}-\mathrm{i}\sum_{j=1}^d\left(\Omega_{kj}a_j+\kappa_{kj} a_j^{\dagger}\right)\\
& + &\frac{1}{2}\sum_{\ell=1}^m\left\{\left(\sum_{j=1}^d v_{\ell j}a_j^{\dagger}+\overline{u}_{\ell j}a_j\right)u_{\ell k}-v_{\ell k}\left(\sum_{j=1}^d\overline{v}_{\ell j}a_j+u_{\ell j}a_j^{\dagger}\right)\right\}\\
\end{eqnarray*}
Using the last equality and $\pounds (a_k^{\dagger})=\pounds (a_k)^*$ concludes the proof.
\end{proof}

\smallskip
The following formula is verified for any generator $\pounds$ of a QMS:
\begin{lemma}\label{lemmaXY}
For all $X,Y\in\mathcal{B}(\mathsf{h})$
\begin{equation}\label{eq:LXY}
\pounds(XY)= X\pounds(Y)+\pounds(X) Y
+\sum_{\ell=1}^m [L_\ell,X^*]^*[L_\ell,Y].
\end{equation}
\end{lemma}
\begin{proof}
Let $\pounds_0(X)=-(1/2)\sum_{\ell=1}^m\left(L_\ell^*L_\ell X-2L_\ell^*XL_\ell+XL_\ell^*L_\ell\right)$ then
we can write \linebreak$\pounds_0(XY)-X\pounds_0(Y)-\pounds_0(X)Y$ as
\begin{eqnarray*}
& &\sum_{\ell=1}^m\left(-\frac{1}{2}L_\ell^*L_\ell XY+L_\ell^*XYL_\ell-\frac{1}{2}XYL_\ell^*L_\ell\right)\\
&+&\sum_{\ell=1}^m\left(\frac{1}{2}XL_\ell^*L_\ell Y-XL_\ell^*YL_\ell+\frac{1}{2}XYL_\ell^*L_\ell\right)\\
&+&\sum_{\ell=1}^m\left(\frac{1}{2}L_\ell^*L_\ell XY-L_\ell^*XL_\ell Y+\frac{1}{2}XL_\ell^*L_\ell Y\right) \\
& = & \sum_{\ell=1}^m \left( L_\ell^*XYL_\ell + XL_\ell^*L_\ell Y
-XL_\ell^*YL_\ell - L_\ell^*XL_\ell Y\right)
\end{eqnarray*}
It follows that $\pounds_0(X)$ is equal to
\[
\sum_{\ell=1}^m \left([L_\ell^*,X]YL_\ell
-[L_\ell^*,X]L_\ell Y \right)  \\
= \sum_{\ell=1}^m \left(-[L_\ell^*,X]\, [L_\ell,Y] \right)
= \sum_{\ell=1}^m \left([L_\ell,X^*]^*\, [L_\ell,Y] \right).
\]
Recalling the usual commutator property $[H,XY]=[H,X]Y+X[H,Y]$,
we find then
\[
\pounds(XY)-X\pounds(Y)-\pounds(X)Y
= \pounds_0(XY)-X\pounds_0(Y)-\pounds_0(X)Y
= \sum_\ell \left([L_\ell,X^*]^*\, [L_\ell,Y] \right).
\]
This completes the proof.
\end{proof}

As a final step towards proving conservativity via Theorem \ref{th:cons}, we prove the following

\begin{proposition} \label{prop:CupperBound}
Let $C=\sum\limits_{k=1}^d a_ka_k^{\dagger}$. There exist a constant $b>0$ such that
$\pounds (C)\leq b C$
\end{proposition}
\begin{proof}
By Lemmas $\ref{lemmaak}$ we have that
\[
\pounds (a_k)  = \sum_{j=1}^d \left(w_{kj}a^{\dagger}_j+z_{kj}a_j\right)-\frac{\mathrm{i}\zeta_k}{2}\mathbbm{1}, \quad
\pounds (a_k^{\dagger}) = \sum\limits_{j=1}^d \left(\overline{w}_{kj}a_j+\overline{z}_{kj}a_j^{\dagger}\right)
+\frac{\mathrm{\mi}\overline{\zeta}_k}{2}\mathbbm{1}
\]
for some complex numbers $w_{kj},z_{kj}, \zeta_j$. While, by Lemma \ref{lemmaXY}, we get
\begin{eqnarray*}
\pounds(a_k a_k^{\dagger}) & = & -\mathrm{i}\overline{\zeta}_k a_k/2
+ \sum_{j=1}^d\left(\overline{w}_{kj}a_ka_j+\overline{z}_{kj}a_ka_j^{\dagger}\right)\\
& + & \mathrm{i}\zeta_k a^\dagger_k/2 +
\sum_{j=1}^d\left(w_{kj}a_j^{\dagger}a_k^{\dagger}+z_{kj}a_ja_k^{\dagger}-2\mathrm{i}\zeta_j a_k^{\dagger}\right)
+\Vert v_{\bullet k}  \Vert^2\mathbbm{1}\\
 & = & \frac{\mathrm{i}}{2}\left( \zeta_k a_k^{\dagger}-\overline{\zeta}_ka_k\right) +
 \sum_{j=1}^d\left(\overline{z}_{kj}a_ka_j^{\dagger}+z_{kj}a_ja_k^{\dagger}
 +\overline{w}_{kj}a_ka_j+w_{kj}a_j^{\dagger}a_k^{\dagger}\right)+\Vert v_{\bullet k}  \Vert^2\mathbbm{1}
\end{eqnarray*}
Note that for each $k,j$
\[
|a_j^{\dagger}-z_{kj}a_k^{\dagger}|^2=a_ja_j^{\dagger}+|{z}_{jk}|^2 a_ka_k^{\dagger}
-\overline{z}_{kj}a_ka_j^{\dagger}-z_{kj}a_ja_k^{\dagger}\geq 0,
\]
it follows that
\begin{equation}\label{ineqza}
\overline{z}_{kj}a_ka_j^{\dagger}+z_{kj}a_ja_k^{\dagger}\leq a_ja_j^{\dagger}+|{z}_{jk}|^2 a_ka_k^{\dagger}
\end{equation}
and, in the same way
\begin{equation}\label{ineqwa}
\overline{w}_{kj}a_ka_j+w_{kj}a_k^{\dagger}a_j\leq |w_{kj}|^2a_ka_k^{\dagger}+a_ja_j^{\dagger}.
\end{equation}
Finally, from
\[
|a_k^{\dagger}+\mathrm{i}\overline{\zeta}_k\mathbbm{1}|^2
=a_ka_k^{\dagger}+|\zeta_k|^2\mathbbm{1}-\mathrm{i}\zeta_k a_k^{\dagger}+\mathrm{i}\overline{\zeta}_k a_k\geq0
\]
it follows that
\begin{equation}\label{ineqzeta}
\mathrm{i}\zeta_ja_k^{\dagger}-\mathrm{i}\overline{\zeta}_ja_k\leq a_ka_k^{\dagger}+|\zeta_j|^2\mathbbm{1}.
\end{equation}
By (\ref{ineqza}),(\ref{ineqwa}), and (\ref{ineqzeta})
\begin{eqnarray}
\pounds(C)\leq \left(3 d\max_{1\leq k\leq d}\left(1+\sum_{j=1}^d(|z_{kj}|^2+|w_{kj}|^2)\right)C
+\sum_{j=1}^d\left(|\zeta_j|^2+\Vert v_{\bullet j}  \Vert^2\right)\mathbbm{1}\right)
\end{eqnarray}
since $C\geq d\mathbbm{1}$ then
$\pounds(C)\leq b C$ with $$b=\max\left\{3d\max_{1\leq k\leq d}\left(\sum\sum_{j=1}^d(|z_{kj}|^2
+|w_{kj}|^2)\right), \sum_{j=1}^d\left(|\zeta_j|^2+\Vert v_{\bullet j}  \Vert^2\right)\,\right\}.$$
\end{proof}

We can eventually state the result on conservativity.
\begin{theorem}
The minimal QDS semigroup generated by the pre-generator \eqref{eq:GKLS} with $H$, $L_\ell$ given by \eqref{eq:H},\eqref{eq:Lell}
is Markov.
\end{theorem}
\begin{proof}
	We apply Theorem \ref{th:cons} with the operator $C$ given by
\begin{equation*}
\hbox{\rm Dom}(C) = \left\{
u=\sum_\alpha u_\alpha e(\alpha) \,\mid\,
\sum_\alpha |\alpha|^2 |u_\alpha|^2 < \infty\,\right\}, \qquad
Cu =  \sum_{j=1}^d a_j a_j^\dagger u .
\end{equation*}
Conditions $1$ and $2$ are satisfied by definition and by Proposition \ref{prop:Ggenerates}. For condition $3$ one can choose
$D$ as in the proof of Proposition \ref{prop:Ggenerates}. In this way $(a)$ is easily satisfied, while $(c)$ follows from
Proposition \ref{prop:CupperBound}. The operator $\Phi$ is the self-adjoint extension of $\sum_\ell L_\ell^*L_\ell$
(defined on the domain $D$) and is second-order polynomial in $a, a^\dagger$. Inequalities like \eqref{ineqwa} and \eqref{ineqza}
allows one to show that $\Phi$ is majorized by a suitable multiple of $C$.
Replacing $C$ with this suitable multiple the proof is completed by Proposition \ref{prop:CupperBound}.
\end{proof}

\subsection{Proof of Theorem \ref{th:explWeyl}}\label{ssect:explWeyl}

Following the proof of Theorem 2 in \cite{AFP} let us rewrite the GKLS pre-generator as
	\[
		\mathcal{L}(W(z)) = \mi \comm{H}{W(z)}
         +\frac{1}{2} \sum_{\ell=1}^m \left( L_\ell^* \comm{W(z)}{L_\ell} + \comm{L_\ell^*}{W(z)}L_\ell \right).
	\]
Recalling that $\Omega=\Omega^*, \kappa=\kappa^{\hbox{\rmseven T}}$ and from (\ref{eq:WzCa}) one gets
	\[
		\comm{W(z)}{L_\ell} = - W(z) \left( \conj{V}z + U\conj{z} \right)_\ell, \quad \comm{L_\ell^*}{W(z)}
            = W(z)\left( V\conj{z} + \conj{U}z\right)_\ell
	\]
and so
	\begin{align*}
	\comm{H}{W(z)} &= W(z)\left[ a\left( \Omega z + \kappa \conj{z} \right) + a^\dagger \left( \Omega z + \kappa \conj{z} \right)
+\frac{1}{2} \scalare{z}{\Omega z + \kappa \conj{z}} + \frac{1}{2} \conj{\scalare{z}{\Omega z + \kappa \conj{z}}} + \rescalar{\zeta}{z}\right]
	\end{align*}
	At last one finds
	\begin{align*}
		\sum_{\ell=1}^m\comm{L_\ell^*}{\comm{W(z)}{L_\ell}}
        &= -\sum_{\ell=1}^mW(z) \left(\conj{V}z + U\conj{z}\right)_\ell \left(V\conj{z} + \conj{U}z\right)_\ell \\
		&= -W(z)\left( \scalare{z}{\conj{V^*V} z + V^{\hbox{\rmseven T}}U \conj{z}} + \conj{\scalare{z}{\conj{U^*U}z
          + U^{\hbox{\rmseven T}}V \conj{z}}}\right) .
	\end{align*}
	Using the previous results one finds that $\mathcal{L}(W(z)) = W(z) X(z)$ for some operator $X(z)$ which is explicitly given by
	\begin{align*}
		X(z) = &a^\dagger \left( \left( \frac{\conj{U^*U - V^*V}}{2} + \mi \Omega \right)z
           + \left( \frac{ U^{\hbox{\rmseven T}}V - V^{\hbox{\rmseven T}} U}{2} + \mi \kappa \right) \conj{z} \right)\\
		&-a \left( \left( \frac{\conj{U^*U - V^*V}}{2} +\mi \Omega\right)z + \left( \frac{U^{\hbox{\rmseven T}}V-V^{\hbox{\rmseven T}}U}{2}
           + \mi \kappa\right)\conj{z} \right)  \\
		&+ \frac{1}{2} \scalare{z}{\mi\Omega z + \mi\kappa \conj{z}} - \frac{1}{2} \conj{\scalare{z}{\mi\Omega z + \mi\kappa \conj{z}}}
         + \mi\rescalar{\zeta}{z} \\
		&- \frac{1}{2} \left( \scalare{z}{\conj{V^*V} z + V^{\hbox{\rmseven T}}U \conj{z}}
       + \conj{\scalare{z}{\conj{U^*U}z + U^{\hbox{\rmseven T}}V \conj{z}}}\right).
	\end{align*}
	Let's derive \eqref{eq:explWeyl} at $t=0$. Using
	\[
		\frac{\text{d}}{\text{dt}} W(\hbox{\rm e}^{tZ})= W(z) \sum_{j=1}^d \left( (Zz)_j a_j^\dagger - (\conj{Zz})_ja_j + \frac{1}{2} \left(\conj{z_j}(Zz)_j - (\conj{Zz})_jz_j \right) \right)
	\]
(with respect to the norm topology) one again has $\mathcal{L}(W(z))= W(z) Y(z)$, where
	\[
		Y(z)= \sum_{j=1}^d\left( (Zz)_j a_j^\dagger - (\conj{Zz})_ja_j + \frac{1}{2}  \left(\conj{z_j}(Zz)_j - (\conj{Zz})_jz_j \right) \right)- \frac{1}{2}\rescalar{z}{Cz} + \mi \rescalar{\zeta}{z}.
	\]
	Since $X(z)$ and $Y(z)$ must coincide for every $z\in \mathbb{C}^d$ the proof of Theorem \ref{th:explWeyl} is complete.

\section*{Appendix B: Characterization of $\mathcal{N}(\mathcal{T})$ } \label{sect:charNT}

In this section we derive the characterization of $\mathcal{N}(\T)$ in terms of iterated commutators.
We begin by illustrating the idea of the proof in the case where the operators $L_\ell$ and $H$ are
bounded. For all $x,y\in\mathcal{B}(\mathsf{h})$, recall the formula (\ref{eq:LXY}) from Lemma \ref{lemmaXY}.
Note that, if $x\in\mathcal{N}(\mathcal{T})$ and $y$ is arbitrary then, since
$\mathcal{T}_t(y^*x)=\mathcal{T}_t(y^*)\mathcal{T}_t(x)$ by Proposition \ref{prop:struct-NT} 1,
taking the derivative at $t=0$ we get
$\mathcal{L}(y^*x) = \mathcal{L}(y^*)x+ y^*\mathcal{L}(x)$,
therefore
\begin{equation}\label{eq:CdC}
\sum_{\ell=1}^m[L_\ell,y]^*[L_\ell, x] =0.
\end{equation}
If the operators $L_\ell$ are bounded, we are allowed to take $x=y$, then $[L_\ell,x]=0$ for all $\ell$.
Moreover, since $x^*$ also belongs to $\mathcal{N}(\mathcal{T})$, taking the adjoint of $[L_\ell,x^*]=0$,
$x$ also commutes with all the operators $L_\ell^*$ and $\mathcal{L}(x)=\mi [H,x]$.
Clearly, since $\mathcal{N}(\mathcal{T})$ is $\mathcal{T}_t$-invariant, $\mathcal{L}(x)=\lim_{t\to 0}(\mathcal{T}_t(x)-x)/t$
belongs to $\mathcal{N}(\mathcal{T})$. Therefore $[L_\ell,[H,x]]=0$ for all $\ell$ and, by the Jacobi identity
\[
\left[\,x,[H,L_\ell]\,\right] = -\left[\,H,[L_\ell,x]\,\right] - \left[\,L_\ell,[x,H]\,\right] = 0.
\]
In this way, one can show inductively that $x$ commutes with the iterated commutators (\ref{eq:genNT}).

If the operators $L_\ell,H$ are unbounded, one has to cope with several problems. The operator $\mathcal{L}$ is
unbounded and, even if we choose $x,y$ in the domain of $\mathcal{L}$, it is not clear whether $y^*x$
belongs to the domain of $\mathcal{L}$ (see \cite{FF00}). Multiplication of generalized commutators $[L_\ell,y]$
$[L_\ell,x]$ may not be defined. If we choose a ``nice'' $y\in\hbox{\rm Dom}(\mathcal{L})$ then it is not
clear whether we can take $x=y$ because we do not know a priori if our ``nice'' $y$ belongs to $\mathcal{N}(\mathcal{T})$.

\smallskip
We begin the analysis of $\mathcal{N}(\mathcal{T})$ by a few preliminary lemmas.

\begin{lemma}\label{lem:derNnorm}
The following derivative exists with respect norm topology for all $z\in\mathbb{C}$
\[
\frac{\rm d}{\hbox{\rm d}t}\mathcal{T}_t(W(z))e_g\Big|_{t=0} =
G^* W(z)e_g + \sum_{\ell=1}^m L_\ell^* W(z) L_\ell e_g + W(z) G e_g
\]
\end{lemma}

\begin{proof}
The right-hand side operator $G^* W(z) + \sum_{\ell=1}^m L_\ell^* W(z) L_\ell  + W(z) G$ is unbounded
(for $z\not=0$) therefore $W(z)$ does not belong to the domain of $\mathcal{L}$ but we can consider
the quadratic form $\pounds(W(z))$ on $D\times D$.
Differentiability of functions $t\mapsto \left\langle \xi',\mathcal{T}_t(x)\xi\right\rangle$ also
holds for $\xi,\xi'$ in the linear span of exponential vectors. Therefore, for all such $\xi$, we have (Theorem \ref{th:G-QMS!})
\[
 \left\langle \xi, \left(\mathcal{T}_t(W(z))-W(z)-t\pounds(W(z))\right)e_g \right\rangle
  =  \int_0^t  \left\langle \xi,  \pounds(\mathcal{T}_s(W(z)) - W(z))  e_g \right\rangle \hbox{\rm d}s.
\]
Recalling that $\mathcal{T}_s(W(z)) =\varphi_z(s)W(\hbox{\rm e}^{sZ}z) $ as in (\ref{eq:explWeyl}) for a complex valued
function $\varphi$  such that $\lim_{s\to 0} \varphi_z(s)=1$, the right-hand side integrand can be written as
\[
\left(\varphi_z(s)-1\right) \left\langle \xi, \pounds(W(\hbox{\rm e}^{sZ}z)) e_g \right\rangle
+  \left\langle \xi, \pounds(W(\hbox{\rm e}^{sZ}z)-W(z)) e_g \right\rangle
\]
A long but straightforward computation shows the function
\[
s\mapsto  \pounds(W(\hbox{\rm e}^{sZ}z)-W(z)) e_g
=\left(G^*W(\hbox{\rm e}^{sZ}z) + \sum_{\ell=1}^m L_\ell^* W(\hbox{\rm e}^{sZ}z) L_\ell + W(\hbox{\rm e}^{sZ}z)G\right)e_g
\]
is continuous vanishing at $s=0$ and the function $s\mapsto  \pounds(W(\hbox{\rm e}^{sZ}z)) e_g$ is
bounded with respect to the Fock space norm. Therefore, taking suprema for $\xi\in\Gamma(\mathbb{C}^d)$,
$\Vert \xi\Vert=1$, we find the inequalities
\begin{eqnarray*}
\left\Vert \left(\mathcal{T}_t(W(z))-W(z)-t\pounds(W(z))\right)e_g\right\Vert
& \leq & \int_0^t  \left|\varphi_z(s)-1\right|\left\Vert  \pounds(W(\hbox{\rm e}^{sZ}z)) e_g\right\Vert \hbox{\rm d}s\\
& + & \int_0^t   \left\Vert \pounds(W(\hbox{\rm e}^{sZ}z)-W(z)) e_g\right\Vert \hbox{\rm d}s
\end{eqnarray*}
The conclusion follows dividing by $t$ and taking the limit as $t\to 0^+$.
\end{proof}


\begin{lemma}\label{lem:weakCdC}
Let $x\in\mathcal{N}(\mathcal{T})$. For all exponential vectors $e_g,e_f$ and
all Weyl operators $W(z)$ we have
\begin{equation}\label{eq:weakCdC}
\sum_{\ell=1}^m\left\langle  \left[L_\ell, W(-z)\right]e_g,  x L_\ell\, e_f\right\rangle
= \sum_{\ell=1}^m\left\langle  L_\ell^*\left[L_\ell, W(-z)\right]e_g,  x \, e_f\right\rangle.
\end{equation}
\end{lemma}
\begin{proof}
If $x\in\mathcal{N}(\mathcal{T})$, then, for all $g,f,z\in\mathbb{C}^d$ and $t\geq 0$ we have
\begin{eqnarray*}
\left\langle e_g,\left(\mathcal{T}_t(W(z)x)-W(z)x\right)e_f\right\rangle & = &
\left\langle e_g,\left(\mathcal{T}_t(W(z))\mathcal{T}_t(x)-W(z)x\right)e_f\right\rangle \\
& = & \left\langle  e_g,\left( \mathcal{T}_t(W(z))-W(z)\right)x\,e_f\right\rangle
+  \left\langle  e_g, W(z) \left( \mathcal{T}_t(x)-x\right)e_f\right\rangle \\
& + & \left\langle \left( \mathcal{T}_t(W(-z))-W(-z)\right) e_g, \left( \mathcal{T}_t(x)-x\right)e_f\right\rangle
\end{eqnarray*}
By Lemma \ref{lem:derNnorm} means  the norm limit
\[
\lim_{t\to 0^+}\frac{ (\mathcal{T}_t(W(-z))-W(-z))e_g}{t} =
\left(G^* W(-z) + \sum_\ell L_\ell^* W(-z) L_\ell + W(-z)G\right)e_g
\]
exists, therefore $\sup_{t>0}t^{-1}\Vert \mathcal{T}_t(W(-z))-W(-z))e_g\Vert <+\infty$. Moreover
\begin{eqnarray*}
   \left\Vert \left( \mathcal{T}_t(x)-x\right)e_f\right\Vert^2 &=&
   \left\langle e_f, \left( \mathcal{T}_t(x^*)-x^*\right)\left( \mathcal{T}_t(x)-x\right) e_f \right\rangle \\
   &\le & \left\langle e_f, \left( \mathcal{T}_t(x^*x)- x^*\mathcal{T}_t(x) -\mathcal{T}_t(x^*)x+x^*x\right) e_f \right\rangle
\end{eqnarray*}
which tends to $0$ as $t\to 0^+$ by weak$^*$ continuity of $\mathcal{T}_t$. As a result
\[
\lim_{t\to 0^+} t^{-1}\left\langle  e_g,\left( \mathcal{T}_t(W(z))-W(z)\right) \left( \mathcal{T}_t(x)-x\right)e_f\right\rangle=0,
\]
therefore
\begin{eqnarray*}
\lim_{t\to 0^+}t^{-1}\left\langle e_g,\left(\mathcal{T}_t(W(z)x)-W(z)x\right)e_f\right\rangle
& = & \lim_{t\to 0^+}t^{-1}\left\langle \left( \mathcal{T}_t(W(-z))-W(-z)\right) e_g,x\,e_f\right\rangle \\
& + &  \lim_{t\to 0^+}t^{-1} \left\langle  e_g, W(z) \left( \mathcal{T}_t(x)-x\right)e_f\right\rangle
\end{eqnarray*}
and we get
\begin{eqnarray*}
 & & \left\langle G e_g,W(z)x\, e_f\right\rangle
+\sum_{\ell=1}^m \left\langle L_\ell e_g, W(z)x L_\ell e_f\right\rangle
+\left\langle e_g, W(z)x\, Ge_f\right\rangle \\
 &=& \left\langle \left(G^* W(-z) + \sum_\ell L_\ell^* W(-z) L_\ell + W(-z)G\right)e_g,x\,e_f\right\rangle \\
 & + & \left\langle  GW(-z)e_g,  x e_f\right\rangle +\sum_{\ell=1}^m\left\langle  L_\ell W(-z)e_g,  x L_\ell e_f\right\rangle
 +\left\langle  W(-z)e_g,  x G e_f\right\rangle.
\end{eqnarray*}
The first term in the left-hand side cancels with the third term in right-hand side and last terms in both sides
cancel as well. Noting that
\[
G^* W(-z)e_g+GW(-z)e_g=-\sum_\ell L_\ell^*L_\ell W(-z)e_g
\]
adding the first and fourth terms in the right-hand side, we find
\begin{eqnarray*}
   \sum_{\ell=1}^m \left\langle L_\ell e_g, W(z)x L_\ell e_f\right\rangle
   &=& - \sum_{\ell=1}^m \left\langle  L_\ell^*L_\ell W(-z)e_g,  x e_f\right\rangle \\
    &+ & \sum_{\ell=1}^m \left\langle L_\ell^* W(-z) L_\ell e_g,x\,e_f\right\rangle
    +\sum_{\ell=1}^m\left\langle  L_\ell W(-z)e_g,  x L_\ell e_f\right\rangle.
\end{eqnarray*}
Rearranging terms we get (\ref{eq:weakCdC}) which is a weak form of identity (\ref{eq:CdC}).
\end{proof}

\smallskip
The following lemma serves to get (\ref{eq:weakCdC}) for each $\ell$ fixed without summation,
taking advantage of the arbitrarity of $z$.
\begin{lemma}\label{lem:VUrange}
For all $\ell_\bullet\in\{\, 1,2,\dots, d\,\}$ fixed there exists $z\in \mathbb{C}^d$ such that
\[
\sum_{j=1}^d \left( \overline{v}_{i j}z_j+u_{i j} \overline{z}_j\right) = \delta_{i, \ell_\bullet}
=\left\{\, \begin{array}{cl}  1 & \hbox{\it if\ } \ i= \ell_\bullet \\ 0 & \hbox{\it if\ } \ i\not= \ell_\bullet \end{array} \right.
\]
\end{lemma}
\begin{proof}
Note that $\overline{V}z+U\overline z$ arises from the map composition
\[
\begin{array}{cccc}
  \quad J_c  &       &   \left[\,\overline{V}\,|\,U\,\right]  & \\
z \ \rightarrow & \left[\begin{array}{c} z \\ \overline{z}\end{array} \right] &  \longrightarrow & \overline{V}z+U\overline z \\
\end{array}
\]
Let $(\phi_\ell)_{1\le \ell\le m}$ be an orthonormal basis of $\mathbb{C}^m$. We look for a $z\in\mathbb{C}^m$ solving the
real linear system
\[
\left[\,\overline{V}\,|\,U\,\right] J_c\, z = \phi_{\ell_\bullet}
\]
Since
\[
\hbox{\rm Ran}\left(\left[\overline{V}\,|\,U\,\right] J_c\right)
= \hbox{\rm Ker}\left(\left[\overline{V}\,|\,U\,\right] J_c)^{\hbox{\rmseven T}}\right)^\perp
= \hbox{\rm Ker}\left(J_c^{\hbox{\rmseven T}}\left[\overline{V}\,|\,U\,\right]^{\hbox{\rmseven T}}\right)^\perp,
\]
$J_c$ is one-to-one and, by the minimality assumption (\ref{eq:minimalCond})
\[
\hbox{\rm Ker}\left(\left[\overline{V}\,|\,U\,\right]^{\hbox{\rmseven T}}\right)
= \hbox{\rm Ker}\left(\left[\begin{array}{c}V^*\\  U^{\hbox{\rmseven T}}\end{array}\right]\right)
=\hbox{\rm Ker}\left(V^*\right)\cap \hbox{\rm Ker}\left( U^{\hbox{\rmseven T}}\right)=\{0\},
\]
we find $\hbox{\rm Ran}\left(\left[\overline{V}\,|\,U\,\right] J_c\right)=\mathbb{C}^m$ and the proof is complete.
\end{proof}

\begin{proposition}\label{prop:NT_in_Lprime}
The decoherence-free subalgebra $x\in\mathcal{N}(\mathcal{T})$ is contained in the generalized commutant
of the operators $L_\ell,L_\ell^*$\, $1\le\ell\le m$.
\end{proposition}

\begin{proof}
For a Weyl operator $W(z)$ we have
\[
[L_\ell,W(z)] = \sum_{j=1}^d\left[\overline{v}_{\ell j} a_j + u_{\ell j}a^\dagger_j, W(z)\right]
= \sum_{j=1}^d\left( \overline{v}_{\ell j}z_j+u_{\ell j} \overline{z}_j\right) W(z).
\]
and (\ref{eq:weakCdC}) becomes
\[
\sum_{\ell=1}^m\sum_{j=1}^d\left( \overline{v}_{\ell j}z_j+u_{\ell j} \overline{z}_j\right)
\left\langle  W(-z)e_g,  x L_\ell e_f\right\rangle
= \sum_{\ell=1}^m\sum_{j=1}^d\left( \overline{v}_{\ell j}z_j+u_{\ell j} \overline{z}_j\right)
 \left\langle  L_\ell^* W(-z) e_g,  x e_f\right\rangle.
\]
By Lemma \ref{lem:VUrange}, by choosing some special $z_\ell\in\mathbb{C}^d$ we find
\[
\left\langle  W(-z_\ell)e_g,  x L_\ell e_f\right\rangle
=  \left\langle  L_\ell^* W(-z_\ell) e_g,  x e_f\right\rangle
\]
for all $g,f\in\mathbb{C}^d$ and all $\ell$. Therefore, by the arbitrarity of $g$ and the explicit action
of Weyl operators on exponential vectors
\[
\left\langle   e_w,  x L_\ell e_f\right\rangle
=  \left\langle  L_\ell^*  e_w,  x e_f\right\rangle
\]
for all $w,f\in\mathbb{C}^d$ and all $\ell$. Since exponential vectors form a core for $L_\ell^*$
and $L_\ell$ is closed, this implies that $x e_f$ belongs to the domain of $L_\ell$ and $L_\ell x e_f = xL_\ell e_f$,
namely $xL_\ell \subseteq L_\ell x$.

Replacing $x$ with $x^*$ we find $x^*L_\ell \subseteq L_\ell x^*$ and standard results on the adjoint of
products of operators (see e.g. \cite{Kato} 5.26 p. 168) lead to the inclusions
\[
x L_\ell^*  \subseteq  \left( L_\ell x^*\right)^* \subseteq  \left( x^*L_\ell\right)^* =L_\ell x^*.
\]
It follows that $x$ belongs to the generalised commutant of the set $\left\{\,L_\ell,L_\ell^*\,\mid\, 1\le \ell\le  m\,\right\}$.
\end{proof}

\begin{lemma}\label{lem:Nn-eiH-inv}
The domain $\operatorname{Dom}(N^{n/2})$ is $\mathrm{e}^{\mi t H}$-invariant for all $t\in\mathbb{R}$
and there exists a constant $c_n>0$ such that
\begin{equation}\label{eq:NeitH}
\left\Vert (N+\unit)^{n/2}  \mathrm{e}^{\mi t H} \xi \right\Vert^2
\leq  \mathrm{e}^{c_n |t|} \left\Vert (N+\unit)^{n/2} \xi\right\Vert^2
\end{equation}
for all $\xi\in \operatorname{Dom}(N^{n/2})$.
\end{lemma}
\begin{proof}
Consider Yosida approximations of the identity operator $(\unit+\epsilon (N+\unit) )^{-1}$ for all $\epsilon>0$
and bounded approximations $X_\epsilon=(N+\unit)^n(\unit+\epsilon (N+\unit) )^{-n}$ of the $n$-the power of $N+\unit$.
Note that, the domain $D$ is invariant for these operators and also $H$ invariant. For all $u\in D$, setting $v_\epsilon = (\unit+\epsilon N )^{-n}u$ we have
\[
  \left\langle u, (X_\epsilon H - H X_\epsilon)u \right\rangle
   = \left\langle v_\epsilon,
    ((N+\unit)^n  H (\unit+\epsilon N )^n \kern-1truept - \kern-1truept  (\unit+\epsilon N )^n H (N+\unit)^n)v_\epsilon \right\rangle
\]
Compute
\begin{align*}
& \left((N+\unit)^n  H (\unit+\epsilon (N+\unit) )^n - (\unit+\epsilon (N+\unit) )^n H (N+\unit)^n\right) \\
& =   \sum_{k=0}^n\binom{n}{k} \epsilon^k ((N+\unit)^n H (N+\unit)^k - (N+\unit)^k H (N+\unit)^n) \\
& =   \sum_{k=0}^n\binom{n}{k} \epsilon^k (N+\unit)^k [(N+\unit)^{n-k}, H\,] (N+\unit)^k
\end{align*}
and noting that the commutator $[(N+\unit)^{n-k}, H]$ is a polynomial in $a_j,a^\dagger_k$ of order $2(n-k)$. This implies that
we can find a constant $c_n$ such that $\left|\langle u', [(N+\unit)^{n-k},H]u'\rangle\right|\leq c_n \Vert (N+\unit)^{(n-k)/2} u'\Vert^2$
(for $u'\in D$) and so we get the inequality
\begin{align*}
  \left|\left\langle u, X_\epsilon H u \right\rangle - \left\langle Hu, X_\epsilon u \right\rangle\right|
   &\leq  c_n \sum_{k=0}^n\binom{n}{k} \epsilon^k
   \left\langle v_\epsilon, (N+\unit)^{n+k}v_\epsilon \right\rangle \\
   & =  c_n  \left\langle v_\epsilon, (N+\unit)^n(\unit+\epsilon (N+\unit) )^nv_\epsilon \right\rangle \\
   & =  c_n \left\langle u, X_\epsilon u \right\rangle.
\end{align*}
The above inequality extends to $u\in\operatorname{Dom}(H)$ by density. \\
Now, for all $u\in \operatorname{Dom}(H)$ and $t\geq 0$,  we have
\begin{align*}
   \left\Vert X_\epsilon^{1/2} \mathrm{e}^{\mi t H} u\right\Vert^2
   &-\left\Vert X_\epsilon^{1/2}  u \right\Vert^2 =  \int_0^t
   \frac{\mathrm{d}}{\mathrm{d}s}\left\Vert X_\epsilon^{1/2} \mathrm{e}^{\mi s H} u \right\Vert^2\, \mathrm{d}s \\
   &= \mi
   \int_0^t
   \left(\left\langle \mathrm{e}^{\mi s H} u,  X_\epsilon H \mathrm{e}^{\mi s H} u\right\rangle
   -\left\langle H\mathrm{e}^{\mi s H} u, X_\epsilon \mathrm{e}^{\mi s H} u\right\rangle
   \right)\mathrm{d}s \\
   & \leq  c_n \int_0^t   \left\Vert X_\epsilon^{1/2} \mathrm{e}^{\mi s H}u\right\Vert^2 \mathrm{d}s.
 \end{align*}
Gronwall's lemma implies and a similar argument for $t<0$ yield
\[
\left\Vert X_\epsilon^{1/2}  \mathrm{e}^{\mi t H} u \right\Vert^2
\leq \mathrm{e}^{c_n |t|} \left\Vert X_\epsilon^{1/2} u \right\Vert^2
\]
for all $t\in\mathbb{R}$. Considering $u\in D$ and taking the limit as $\epsilon$ goes to zero we get (\ref{eq:NeitH}) for $\xi\in D$ and, finally
for $\xi\in  \operatorname{Dom}(N^{n/2})$ because $D$ is a core for $N^{n/2}$.
\end{proof}

\begin{lemma}\label{lem:a-ytic}
For all $j$ there exists $M_d(\mathbb{C})$ valued analytic functions $\mathsf{H}^{-},\mathsf{H}^{+}$
such that
\[
\hbox{\rm e}^{-\mi t H}\, a_j\, \hbox{\rm e}^{\mi t H} \xi
= \sum_{k=1}^d \left(\mathsf{H}^{-}_{jk}(t)\,  a_k\,  \xi
+ \mathsf{H}^{+}_{jk}(t)\,  a_k^\dagger  \xi\right)
\]
for all $t\in\mathbb{R}$, $\xi\in\hbox{\rm Dom}(N)$.
\end{lemma}
\begin{proof}
For all $\xi',\xi\in\hbox{\rm Dom}(N)$ we have
\begin{eqnarray*}
  \frac{\hbox{\rm d}}{\hbox{\rm d}t} \left\langle \xi', \hbox{\rm e}^{-\mi t H}\, a_j\, \hbox{\rm e}^{\mi t H}\xi\right\rangle
  &=& \lim_{s\to 0}  \frac{1}{s}\left\langle \xi',  \left(\hbox{\rm e}^{-\mi (t+s) H}\, a_j\, \hbox{\rm e}^{\mi (t+s) H}
  -\hbox{\rm e}^{-\mi t H}\, a_j\, \hbox{\rm e}^{\mi t H}\right)\xi\right\rangle\\
   &=& \lim_{s\to 0}  s^{-1}  \left\langle \left(\hbox{\rm e}^{\mi (t+s) H}- \hbox{\rm e}^{\mi t H}\right)\xi',
  a_j\, \hbox{\rm e}^{\mi t H}\xi\right\rangle  \\
  &+& \lim_{s\to 0}  s^{-1}  \left\langle a_j^\dagger\hbox{\rm e}^{\mi t H}\xi',  \left(\hbox{\rm e}^{\mi (t+s) H}- \hbox{\rm e}^{\mi t H}\right)\xi\right\rangle   \\
  & = & \left\langle \mi H \hbox{\rm e}^{\mi t H}\xi',  a_j\, \hbox{\rm e}^{\mi t H}\xi\right\rangle
  + \left\langle a_j^\dagger\, \hbox{\rm e}^{\mi t H}\xi',  \mi H\hbox{\rm e}^{\mi t H}\xi\right\rangle.
\end{eqnarray*}
Now, for all $u,v\in D$ we have
\[
\left\langle \mi H v,  a_j u\right\rangle   + \left\langle a_j^\dagger v ,  \mi H u \right\rangle
= -\mi \left\langle v,  [H,a_j] u\right\rangle
= \sum_{k=1}^d \left(c^{-}_{jk} \left\langle v,  a_k u\right\rangle + c^{+}_{jk}\left\langle v,  a^\dagger_k u\right\rangle\right)
\]
for some complex constants $c^{-}_{jk},c^{+}_{jk}$. The left-hand and right-hand side make sense for $u,v\in\hbox{\rm Dom}(N)$,
therefore they can be extended by density and so
\[
  \frac{\hbox{\rm d}}{\hbox{\rm d}t} \left\langle \xi', \hbox{\rm e}^{-\mi t H}\, a_j\, \hbox{\rm e}^{\mi t H}\xi\right\rangle
  = \sum_{k=1}^d \left(c^{-}_{jk} \left\langle \xi',  \hbox{\rm e}^{-\mi t H}\, a_k\, \hbox{\rm e}^{\mi t H} \xi\right\rangle
  + c^{+}_{jk}\left\langle \xi',   \hbox{\rm e}^{-\mi t H}\, a^\dagger_k \hbox{\rm e}^{\mi t H} \xi\right\rangle\right)
\]
for all $\xi',\xi\in \hbox{\rm Dom}(N)$. Considering the conjugate we find a differential equation
for $ \left\langle \xi', \hbox{\rm e}^{-\mi t H}\, a_j^\dagger\, \hbox{\rm e}^{\mi t H}\xi\right\rangle$ an so we get
a linear system of $2d$ differential equations with constant coefficients. The solution of the system yields analytic
functions $\mathsf{H}^{-},\mathsf{H}^{+}$ as blocks of the exponential of a $2d\times 2d$ matrix.
\end{proof}

\medskip
\begin{proof} {(of Theorem \ref{th:NT-comm})} Let $G_0$ be the self-adjoint extension of $-\sum_{\ell=1}^d L_\ell^*L_\ell/2$.
By Proposition \ref{prop:NT_in_Lprime}, for all $y\in{\mathcal{N}}(\mathcal{T})$ and all $v,u\in \hbox{\rm Dom}(N)$, we have
\begin{eqnarray*}
& & \left\langle G_0 v, y u \right\rangle
+\sum_{\ell=1}^m \left\langle L_\ell v, y L_\ell u \right\rangle
+ \left\langle v, y G_0 u \right\rangle \\
& = & -\frac{1}{2}\sum_{\ell= 1}^m\left(\left\langle L_\ell^*L_\ell v, y u \right\rangle
-2\left\langle L_\ell v, y L_\ell u \right\rangle
+\left\langle v, y L_\ell^*L_\ell u \right\rangle
\right) \\
& = & -\frac{1}{2}\sum_{\ell= 1}^m\left(
\left\langle L_\ell^* y^* L_\ell v, u \right\rangle
-2\left\langle L_\ell v, y L_\ell u \right\rangle
+\left\langle v, L_\ell^* y L_\ell u \right\rangle
\right) = 0
\end{eqnarray*}
because $L_\ell^* y^*$ and $L_\ell^* y$ are extensions of
$y^*L_\ell^*$ and $y L_\ell^* $ respectively, namely $\pounds(x)=\mi[H,x]$
(as a quadratic form).

Now, recalling that ${\mathcal{N}}(\mathcal{T})$ is $\T_s$-invariant by Proposition \ref{prop:struct-NT} 1.
for all $v,u\in\hbox{\rm Dom}(N)$ also $\hbox{\rm e}^{-\mi(t-s)H}v$ and $\hbox{\rm e}^{-\mi(t-s)H}u$
belong to the domain of $N$ by Lemma \ref{lem:Nn-eiH-inv},  we have
\[
\frac{\hbox{\rm d}}{\hbox{\rm d}s} \left\langle  \hbox{\rm e}^{-\mi(t-s)H}v,
\mathcal{T}_s(x) \hbox{\rm e}^{-\mi(t-s)H}u \right\rangle =0
\]
which implies
\[
\mathcal{T}_t(x) = \hbox{\rm e}^{\mi t H}\, x\, \hbox{\rm e}^{-\mi t H}.
\]
From $\mathcal{T}_t$-invariance of ${\mathcal{N}}(\mathcal{T})$  it follows that
also $\hbox{\rm e}^{\mi tH}\,x\,\hbox{\rm e}^{-\mi tH}$
belongs to the generalized commutant of the operators
$L_\ell, L_\ell^*$ ($\ell\ge 1$).

Since $\hbox{\rm Dom}(N)$ is $\hbox{\rm e}^{\mi t H}$-invariant by Lemma \ref{lem:Nn-eiH-inv}, replacing $\xi\in\hbox{\rm Dom}(N)$
by $\hbox{\rm e}^{\mi t H}\xi\in\hbox{\rm Dom}(N)$ and left multiplying by $\hbox{\rm e}^{-\mi t H}$ the identity
$\hbox{\rm e}^{\mi t H}x\,\hbox{\rm e}^{-\mi t H}L_\ell \xi=L_\ell\, \hbox{\rm e}^{\mi t H} x\, \hbox{\rm e}^{-\mi t H}\xi$
becomes
\[
x\,\hbox{\rm e}^{-\mi t H}L_\ell\, \hbox{\rm e}^{\mi t H}\xi=\hbox{\rm e}^{-\mi t H}L_\ell\, \hbox{\rm e}^{\mi t H} x\, \xi
\]
Taking the scalar product with two exponential vectors and differentiating $n$ times at $t=0$ the identity
  \[
  \left\langle v, x\, \hbox{\rm e}^{-\mi t H}\, L_\ell\, \hbox{\rm e}^{\mi t H}u\right\rangle
  =  \left\langle \hbox{\rm e}^{-\mi t H}\, L_\ell^*\, \hbox{\rm e}^{\mi t H} v, x u \right\rangle
  \]
with $u,v\in\hbox{\rm Dom}(N)$, we get
\[
\left\langle v, x\, \delta^n_H (L_\ell) u\right\rangle
  =  \left\langle \delta^n_H (L^*_\ell) v, x u \right\rangle.
\]
Since iterated commutators $\delta^n_H (L_\ell)$ are first order polynomials in $a_j,a^\dagger_k$, this
means that $x$ belongs to the generalized commutant of $\delta^n_H (L_\ell)$. The same argument applies
to generalized commutators of $\delta^n_H (L_\ell^*)$ for all $\ell\geq 1$, $n\geq 0$.

Conversely, if $x$ belongs to the generalized commutant of operators $\delta^n_H (L_\ell)$,
$\delta^n_H (L_\ell^*)$ for all $\ell\geq 1$, $0 \leq n\leq 2d-1$, recall that each one of these
generalized commutators is a first order polynomial in $a_j,a^\dagger_k$ and so determines two
vectors (coefficients of creation and annihilation operators) $\overline{v},u\in \mathbb{C}^d$ and, eventually,
a vector $[\overline{v},u]^{\hbox{\rmseven T}}\in \mathbb{C}^{2d}$.
Let $\mathcal{V}_n$ be the complex linear subspace of $\mathbb{C}^{2d}$ determined by vectors in
$\mathbb{C}^{2d}$ corresponding to generalized commutators of order less or equal than $n$.
Clearly,  $\mathcal{V}_n\subseteq \mathcal{V}_{n+1}$ for all $n$ and so the dimensions dim$_{\mathbb{C}}(\mathcal{V}_n)$
form a non decreasing sequence of natural numbers bounded by $2d$. Moreover, if dim$_{\mathbb{C}}(\mathcal{V}_n)=$dim$_{\mathbb{C}}(\mathcal{V}_{n+1})$,
then $\mathcal{V}_n=\mathcal{V}_{n+1}$ and so
\begin{eqnarray*}
\delta^{n+1}_H (L_\ell) & = & \sum_{m=0}^n \left(z_m\delta_H^{(m)}(L_\ell) +w_m\delta_H^{(m)}(L_\ell^*) \right) + \eta_n\unit, \\
\delta^{n+1}_H (L_\ell^*) & = & \sum_{m=0}^n (-1)^{m}\left(\overline{w}_m\delta_H^{(m)}(L_\ell)
+\overline{z}_m\delta_H^{(m)}(L_\ell^*) \right) + \overline{\eta}_n\unit,
\end{eqnarray*}
for some $z_1,\dots,z_n,w_1,\dots,w_n,\eta_n\in\mathbb{C}$. Iterating, it turns out that the
linear part in creation and annihilation operators of $\delta^{n+m}_H (L_\ell)$ and $\delta^{n+m}_H (L_\ell)$
depends on vectors in $\mathcal{V}_{n}$ for all $m\geq 0$. It follows that, starting from a value $n_0\geq 1$
(corresponding to the zero order commutators $L_\ell$ and $L_\ell^*$), the sequence of dimensions has to increase
at least by $1$ before reaching the maximum value. As a consequence, this is attained in at most $2d-1$ steps.

Summarizing, if $x$ belongs to the generalized commutant of operators $\delta^n_H (L_\ell)$, $\delta^n_H (L_\ell^*)$
for all $\ell\geq 1$, $0 \leq n\leq 2d-1$, then it belongs to generalized commutant of these operators
for all $n\geq 0$. By Lemma \ref{lem:a-ytic}, we can consider the analytic function on $\mathbb{R}$
\[
t\mapsto \left\langle \xi', x\, \hbox{\rm e}^{-\mi tH}L_\ell\,\hbox{\rm e}^{\mi tH}\xi\right\rangle
- \left\langle \hbox{\rm e}^{-\mi tH}L_\ell^*\,\hbox{\rm e}^{\mi tH}\xi', x\, \xi\right\rangle
\]
for all $\xi,\xi'\in D$. The $n$-th derivative at $t=0$ is $(-\mi)^n $ times
\[
\left\langle \xi', x\, \delta^{n}_H (L_\ell)\xi\right\rangle
- \left\langle \delta^{n}_H (L_\ell^*)\xi', x\, \xi\right\rangle  = 0
\]
for all $n\geq 0$ because $x$ belongs to the generalized commutant of operators $\delta^n_H (L_\ell)$.
The same argument shows that $x$ belongs to the generalized commutant of operators $\delta^n_H (L_\ell^*)$.
Applying Theorem 4.1 of \cite{DhFaRe}
(with $C=N$ and keeping in mind that $[G,C],[G^*,C]$ are second order polynomials in
$a_j,a^\dagger_k$, therefore relatively bounded with respect to $C$ whence with respect to $C^{3/2}$)
it follows that $\mathcal{T}_t(x)=\hbox{\rm e}^{\mi tH}x\,\hbox{\rm e}^{-\mi tH}$.

The same conclusion holds for $x^*$ and $x^*x$ because they belong to the generalized commutant of operators
$\delta^n_H (L_\ell),\delta^n_H (L_\ell^*)$. Therefore $x\in \mathcal{N}(\mathcal{T})$ and the proof is
complete.
\end{proof}

\section*{Appendix C: proof of Proposition \ref{prop:sympProp}}

\begin{proof}
Statement 2. readily follows by noticing that if $z \in M_1$ is such that $\Im\langle z, z_1 \rangle=0$
for all $z_1 \in M_1$, then $z \in M_1 \cap {M_1}'$. Therefore $\Im\langle \cdot, \cdot \rangle$ is non-degenerate
when restricted to $M_1$ if and only if $M_1 \cap {M_1}' = \{0 \}$. \\
	We now prove the first one on the existence of the symplectomorphism for $M$. The first step in this proof
is an adaptation of the Gram-Schmidt procedure to symplectic spaces. Consider $z_1 \in M$ with $z_1 \neq 0$ and observe
that there exists $z \in M$ such that $\Im\langle z_1, z \rangle \neq 0$, otherwise the symplectic form would be degenerate on $M$.
We can now set $		z_2 = {z}/{ \Im\langle z_1, z \rangle}$, so that $\Im\langle z_1, z_2 \rangle = 1$.
Let $M_1 = \operatorname{Lin}_\mathbb{R} \{ z_1, z_2 \}$ we now show that $\dim_\mathbb{R} M_1 = 2$ and
	\[
		M = M_1 \oplus {M_1}',
	\]
where both $M_1$ and ${M_1}'$ are symplectic spaces. Clearly if  $z_1, z_2$ were linearly dependent
we would have $z_1 = sz_2$ for some $s \in \mathbb{R}$ and then $\Im \langle z_1, z_2 \rangle = 0$ which contradicts
the construction of $z_2$. Again since $\Im \langle z_1, z_2 \rangle \neq 0$ we have $M_1 \cap {M_1}' = \{ 0 \}$ and
$M_1$ is a symplectic subspace for what we proved at the beginning. Moreover if $z_1 \in M_1$ we can write
	\[
		z = (z + \Im \langle z, z_1 \rangle z_2 - \Im \langle z, z_2 \rangle z_1 ) + (\Im \langle z, z_2 \rangle z_1 - \Im \langle z , z_1 \rangle z_2),
	\]
	where it holds
	\[
		z + \Im \langle z, z_1 \rangle z_2 - \Im \langle z, z_2 \rangle z_1 \in {M_1}', \quad \Im \langle z, z_2 \rangle z_1 - \Im \langle z , z_1 \rangle z_2 \in M_1,
	\]
	hence we have proved $M= M_1 \oplus {M_1}'$.
	Eventually ${M_1}'$ is a symplectic space since it holds ${M_1}'' = M_1$ and ${M_1}'' \cap {M_1}' = \{0\}$. Note also that $\dim_\mathbb{R} {M_1}' \neq 1$ otherwise the symplectic form would be degenerate on it. \\
	We can now repeat the same reasoning starting with the symplectic space ${M_1}'$ in order to obtain $z_3,z_4 \in {M_1}'$ such that $\Im \langle z_3, z_4 \rangle =1$ and if $M_2 = \operatorname{Lin}_\mathbb{R} \{ z_3, z_4 \}$ we have
	\[
		M = M_1 \oplus M_2 \oplus {M_2}',
	\]
	where $M_1, M_2, {M_2}'$ are symplectic spaces with $\dim_\mathbb{R} M_j =2$ and $\dim_\mathbb{R} {M_2}' \neq 1$. Notice that they are also pairwise symplectically orthogonal, since $M_2, {M_2}' \subset {M_1}'$. Since the remainder space ${M_j}'$ has always dimension different from $1$ we can iterate this process until we get ${M_j}'=\{0\}$. When the procedure stops we have
	a sequence $M_1, \dots, M_{d_1}$ of mutually (symplectically) orthogonal symplectic spaces, with $M_j = \operatorname{Lin}_\mathbb{R} \{ z_{2j}, z_{2j+1}\}$, $\Im\langle z_{2j},z_{2j+1} \rangle =1$. Clearly $2d_1 = \dim_\mathbb{R} M$ and this concludes the first step of the proof. \\
	In order to conclude the proof of this point is sufficient to construct the symplectomorphism via
	\[
		B z_{2j} = e_j, \quad Bz_{2j+1} = \mi e_j.
	\]
	Eventually for the symplectomorphism of $M_1 \subset M$ an isotropic subset, consider $\{ z_1, \dots, z_{d_1} \}$ a real linear basis of $M_1$ with $d_1 = \dim_\mathbb{R} M_1$. Since $M_1$ is isotropic we have
	\[
		\Im \langle z_j, z_k \rangle = 0 \quad \forall j,k=1, \dots, d_1.
	\]
This proves 2. In order to prove 3. it suffices to define
	\[
		B z_j = e_j,
	\]
to get a symplectomorphism $B$.
\end{proof}

\section*{Acknowledgements}
We wish to thank Francesco Fidaleo for pointing us out reference \cite{LeRoTe}. \\
This work began when FF was visiting the Department of Mathematics of Escuela Colombiana de Ingenier\'{\i}a ``Julio Garavito''
in July 2019, he would like to thank all the colleagues for the enjoyable atmosphere. \\
The financial support from GNAMPA-INDAM 2020 project ``Processi stocastici quantistici e applicazioni'' is gratefully acknowledged.


\bigskip

Authors' addresses

\bigskip

Juli\'an Agredo

Department of Mathematics

Escuela Colombiana de Ingenier\'{\i}a Julio Garavito,

Bogot\'a, Colombia

{\tt julian.agredo@escuelaing.edu.co}

\medskip

Franco Fagnola

Dipartimento di Matematica,

Politecnico di Milano

Piazza Leonardo da Vinci 32, I-20133 Milano (Italy)

{\tt franco.fagnola@polimi.it}

\medskip

Damiano Poletti

Dipartimento di Matematica,

Universit\`a di Genova

Via Dodecaneso 35, I-16146 Genova (Italy)

{\tt damiano.poletti@gmail.com}


\begin{thebibliography}{100}




\bibitem{AFP}
J. Agredo, F. Fagnola and D. Poletti, Gaussian Quantum Markov Semigroups
on a One-Mode Fock Space: Irreducibility and Normal Invariant States.
{\sl Open Sys. Information Dyn. } {\bf 28} (1) 2150001 (2021).

\bibitem{AFR} J. Agredo, F. Fagnola and R. Rebolledo,
Decoherence free subspaces of a quantum Markov semigroup,
{\sl J. Math. Phys.}  {\bf 55}  112201 (2014).

\bibitem{Araki}
H. Araki, A Lattice of von Neumann Algebras Associated with the Quantum Theory of a Free Bose Field.
{\sl J. Math. Phys.} {\bf 4} 1343--1362 (1963).

\bibitem{BlOl}
Ph. Blanchard and R. Olkiewicz, Decoherence induced transition from quantum to classical dynamics,
{\sl Rev. Math. Phys.} {\bf 15} (3)  217-243 (2003).

\bibitem{BrRo}
O. Bratteli, D.W. Robinson, {\it Operator Algebras and Quantum Statistical Mechanics I}, Springer-Verlag, Berlin 1979.
Second Printing 2002.

\bibitem{CdS}
A. Cannas da Silva, {\it Lectures on Symplectic Geometry}. Corrected 2nd printing 2008.
Springer, Berlin Heidelberg (2008).

\bibitem{CaSaUm}
R. Carbone, E. Sasso and V. Umanit\`a, Environment induced decoherence for Markovian evolutions,
{\sl J. Math. Phys.} {\bf 56} 092704 (2015).

\bibitem{CGMZ}
M. Cattaneo, G.L. Giorgi, S. Maniscalco and R. Zambrini,
Symmetry and block structure of the Liouvillian superoperator in partial secular approximation,
{\it Phys. Rev. A} {\bf 101} 042108 (2020).

\bibitem{ChFa}
A.M.~Chebotarev and F.~Fagnola, Sufficient Conditions for Conservativity of Minimal Quantum Dynamical Semigroups.
 {\sl J. Funct. Anal.} {\bf 153}  no.3  382--404 (1998).




\bibitem{CrFi}
V. Crismale and F. Fidaleo, Symmetries and ergodic properties in
quantum probability. {\sl Colloq. Math.} {\bf 149} 1-20 (2017).

\bibitem{Demoen}
B. Demoen, P. Vanheuverzwijn and A. Verbeure, Completely positive maps on the CCR-algebra.
{\sl Lett. Math. Phys.} \textbf{2} 161 -- 166 (1977).

\bibitem{Derez}
J. Derezi\'nski, Bosonic quadratic Hamiltonians.
{\sl J. Math. Phys.} \textbf{58} 121101 (2017).

\bibitem{DFSU}
J. Deschamps, F. Fagnola, E. Sasso and V. Umanit\`a,
Structure of uniformly continuous quantum Markov semigroups.
{\sl Rev. Math. Phys.} {\bf 28} (1) 1650003  (2016).

\bibitem{DhFaRe}
 A.~Dhahri,  F.~Fagnola and  R.~Rebolledo,
The decoherence-free subalgebra of a quantum Markov
semigroup with unbounded generator,
 {\sl Infin. Dimens. Anal. Quantum Probab. Relat. Top.},
{\bf 13} (3) 413--433  (2010).


\bibitem{Evans}
D.E. Evans, {Irreducible quantum dynamical semigroups},
{\sl Commun. Math. Phys.} {\bf 54} 293--297 (1977).

\bibitem{Fa}
F.~Fagnola, Quantum Markov Semigroups and Quantum Markov Flows. {\sl Proyecciones} {\bf 18}
 n.3 1--144 (1999).

\bibitem{FF00}
F.~Fagnola, A simple singular quantum Markov semigroup. Proceedings of the Third International Workshop
{\sl Stochastic Analysis and Mathematical Physics ANESTOC '98} Birk\"auser 2000 p. 73--88.


\bibitem{FaSaUm}
F. Fagnola, E. Sasso and V. Umanit\`a, The role of the atomic decoherence-free subalgebra in the study of quantum Markov semigroups
{\sl J. Math. Phys.}  {\bf 60}, 072703 (2019).

\bibitem{FGNV}
F. Fagnola, J.E. Gough, H.I. Nurdin and L. Viola,
Mathematical models of Markovian dephasing, {\sl J. Phys. A} {\bf 52} (38) 385301 (2019).

\bibitem{GJN}
J. E. Gough, M. R. James and H. I. Nurdin, Squeezing components in linear quantum feedback networks,
{\sl Phys. Rev. A} {\bf 81} 023804 (2010).

\bibitem{Hi}
P.D. Hislop, A simple proof of duality for local algebras in free quantum field theory
{\sl J. Math. Phys.} {\bf 27} 2542--2550 (1986).

\bibitem{Isar}
A. Isar, Decoherence and asymptotic entanglement in open quantum dynamics,
{\sl J. Russ. Laser Res.} {\bf 28} no. 5 439--452 (2007).

\bibitem{Jo}
P.E.T.~Jorgensen, Approximately Reducing Subspaces for Unbounded Linear Operators. {\sl J. Funct. Anal.} {\bf 23}
392–-141 (1976).

\bibitem{Kato}
T. Kato, \textit{Perturbation theory for linear operators}.
Corrected printing of the second edition.  Springer-Verlag, Berlin, Heidelberg, New York 1980.

\bibitem{LeRoTe}
P. Leyland, J. Roberts and D. Testard, Duality for free quantum Fields,
Centre de Physique Th\'eorique Marseille, Report CPT-78/P-1016, July 1978
{\tt https://inspirehep.net/literature/132161}

\bibitem{Lidar}
D.A. Lidar, Review of decoherence free subspaces, noiseless subsystems, and dynamical decoupling.
{\sl Adv. Chem. Phys} {\bf 154} 295–354 (2014).

\bibitem{LiChWh}
D.A. Lidar, I.L. Chuang, and K.B. Whaley,
Decoherence-free subspaces for quantum computation, {\sl Phys. Rev. Lett.} {\bf 81} 2594 (1998).

\bibitem{Pa}
K. R.~Parthasarathy, {\it An introduction to Quantum Stochastic Calculus}. Birkhauser, Basel, 1992.

\bibitem{Po2021}
D. Poletti, Characterization of Gaussian Quantum Markov Semigroups.
In: Proceedings QP41 (to appear).

\bibitem{ReSi}
M. Reed and B. Simon, {\it Methods of Modern Mathematical Physics. II: Fourier Analysis Self-Adjointness}
Academic Press, San Diego 1975.

\bibitem{Te}
A.E. Teretenkov,  Irreversible quantum evolution with quadratic generator.
{\sl Infin. Dimens. Anal. Quantum Probab. Relat. Top.} {\bf 22}  no. 4 1930001 (2019).

\bibitem{Vheu}
P. Vanheuverzwijn, Generators for quasi-free completely positive semigroups,
{\sl Ann. Inst. H. Poincar\'e Sect. A (N.S.)}, \textbf{29} 123 -- 138 (1978).

\end{thebibliography}
\end{document}